\documentclass[11pt]{article}
\usepackage[utf8]{inputenc}
\usepackage{graphicx}
\graphicspath{ {images/} }
\usepackage{amsmath,amssymb,amsthm,textcomp}
\usepackage{titling}
\usepackage{bbm}
\usepackage{caption}
\usepackage{tikz,pgfplots}
\usepackage[margin=1in]{geometry}
\usepackage{multirow,booktabs}
\usepackage [english]{babel}
\usepackage [autostyle, english = american]{csquotes}
\MakeOuterQuote{"}
\newtheorem{prop}{Proposition}
\newtheorem{theorem}{Theorem}
\newtheorem{corollary}{Corollary}
\newtheorem{definition}{Definition}
\newtheorem{lemma}[theorem]{Lemma}
\usepackage{setspace}
\usepackage{tikz}
\usetikzlibrary{matrix}
\usepackage{enumitem}
\usepackage{authblk}
\usepackage[hyphens,spaces]{url} 
\providecommand{\keywords}[1]
{
  \small	
  \textbf{\textit{Keywords---}} #1
}
\usepackage[colorlinks,citecolor=black,urlcolor=black,bookmarks=false,hypertexnames=true]{hyperref} 

\usepackage{xcolor}
\usepackage{hyperref}
\makeatletter
\def\p@figure{\color{black}}
\def\p@equation{\color{black}}
\def\p@table{\color{black}}
\def\p@section{\color{black}}
\def\p@subsection{\color{black}}
\def\p@subsubsection{\color{black}}
\def\p@prop{\color{black}}
\def\p@theorem{\color{black}}
\def\p@corollary{\color{black}}
\makeatother

\usepackage{natbib}
 \bibpunct[, ]{(}{)}{,}{a}{}{,}%

\title{\textbf{\LARGE{Is More Precise Word of Mouth Better for a High Quality Firm? ... Not Always}}}
\usepackage{geometry}
\author{\vspace{-5ex}}
\author{
\large{Mohsen Foroughifar}\thanks{Rotman School of Management, University of Toronto, mohsen.foroughifar@rotman.utoronto.ca}
\ \ \
 David Soberman\thanks{Rotman School of Management, University of Toronto, david.soberman@rotman.utoronto.ca \\ We would like to thank Heski Bar-Isaac, Peter Landry, Tanjim Hossain, Zachary Zhong, Dina Mayzlin, Nitin Mehta, Mengze Shi, Matthew Osborne, and Ryan Webb for helpful comments. We would also like to thank all participants of the marketing seminar at the Rotman School of Management and the Marketing Science Conference for feedback.} 
}

\begin{document}

\maketitle
\begin{abstract}

Consumers often resort to third-party information such as word of mouth, testimonials and reviews to learn more about the quality of a new product. However, it may be difficult for consumers to assess the precision of such information. We use a monopoly setting to investigate how the precision of third-party information and consumers’ ability to recognize precision impact firm profits. Conventional wisdom suggests that when a firm is high quality, it should prefer a market where consumers are better at recognizing precise signals. Yet in a broad range of conditions, we show that when the firm is high quality, it is more profitable to sell to consumers who do not recognize precise signals. Given the ability of consumers to assess precision, we show a low quality firm always suffers from more precise information. However, a high quality firm can also suffer from more precise information. The precision range in which a high quality firm gains or suffers from better information depends on how skilled consumers are at recognizing precision.
\\
\\
\end{abstract}

\newpage
\section{Introduction}
When a new product is introduced, consumers typically have poor information about product quality. Before purchasing, they often obtain information from third-parties such as friends, online spokespeople, and reports from independent experts. If the information is precise, it can alleviate consumer uncertainty about product quality. But if the information is imprecise, it might not be helpful. Practically, it is straightforward for consumers to understand whether information is in favor of a product or against it. However, it is not always clear how accurate the information is, and that might affect its impact on consumer behavior. In Table \ref{table:examples}, we present several examples in which a consumer obtains third party information and might have difficulty evaluating its precision.

To be specific, there are many sources of information which help consumers to assess product quality. As noted above, word of mouth (WOM) from acquaintances, celebrities who publicly speak about products, and social media posts which discuss a recently launched product are good examples. There is certainly research that examines the effect of information from third-parties, yet the impact of how consumers assess the precision of information from these sources is unexplored. 

Information precision often depends on the source that generates the information. If the person who disseminates the information has significant experience, typically the information is more precise than from a person with limited experience. For example, when a reviewer talks about the quality of a product, whether or not his/her information is precise depends on his/her expertise in that specific product category. Imagine a social media post, created by a reviewer, that promotes a newly introduced digital camera. If this person has a lot of experience with digital cameras, then the post might be quite precise. However, if he has limited experience with digital cameras, his information might be less helpful. The precision of information should have a significant impact on how a market functions.

\begin{table}[ht]
\centering
\small
\caption{Examples where consumers might struggle to evaluate third-party information precision}
\label{table:examples} 
\begin{tabular}{|c|c|c|c|c|}
\hline
\begin{tabular}[c]{@{}c@{}}\textbf{Product} \end{tabular}                                 & \begin{tabular}[c]{@{}c@{}}\textbf{Information} \\ \textbf{source}\end{tabular}             & \begin{tabular}[c]{@{}c@{}}\textbf{When/Where} \\ \textbf{information} \\ \textbf{received}\end{tabular}                               & \begin{tabular}[c]{@{}c@{}}\textbf{Why there are} \\ \textbf{differences} \\ \textbf{ in source} \\  \textbf{precision}\end{tabular}                & \begin{tabular}[c]{@{}c@{}}\textbf{Why there are} \\ \textbf{differences in} \\ \textbf{consumers abilities'} \\ \textbf{to recognize} \\ \textbf{precision}\end{tabular} \\ \hline
Digital camera                                                                              & \begin{tabular}[c]{@{}c@{}}Online \\ spokesperson \end{tabular} & Instagram                                                                                             & \begin{tabular}[c]{@{}c@{}}Does the \\ spokesperson have \\ experience with   \\ photography?\end{tabular} & \begin{tabular}[c]{@{}c@{}}Does the consumer \\ have experience \\ with photography?\end{tabular}                                        \\ \hline
\begin{tabular}[c]{@{}c@{}}New \\ neighbourhood \\ restaurant\end{tabular}                  & \begin{tabular}[c]{@{}c@{}}Next door \\ neighbour\end{tabular}            & \begin{tabular}[c]{@{}c@{}}While cutting \\ the hedge \\ last weekend\end{tabular}                    & \begin{tabular}[c]{@{}c@{}}Is the neighbour \\ a foodie?\end{tabular}                                       & \begin{tabular}[c]{@{}c@{}}How long has \\ the consumer known \\ the neighbour?\end{tabular}                                           \\ \hline
\begin{tabular}[c]{@{}c@{}}Recently opened \\ hotel in Prince \\ Edward County\end{tabular} & \begin{tabular}[c]{@{}c@{}}Independent \\ expert\end{tabular}             & \begin{tabular}[c]{@{}c@{}}Prince Edward \\ County “What’s \\ happening \\ website”\end{tabular}      & \begin{tabular}[c]{@{}c@{}}Has the expert \\ rated new \\ hotels before?\end{tabular}                      & \begin{tabular}[c]{@{}c@{}}Has the consumer \\ looked at hotel \\ reviews before?\end{tabular}                                        \\ \hline
\begin{tabular}[c]{@{}c@{}}Car repair \\ garage \\ near work\end{tabular}                   & \begin{tabular}[c]{@{}c@{}}A friend \\ at work\end{tabular}              & \begin{tabular}[c]{@{}c@{}}At lunch, the\\ consumer told \\him that her \\ car was making \\ a funny noise\end{tabular} & \begin{tabular}[c]{@{}c@{}}Has the friend \\ had his car \\ fixed at different \\ car garages?\end{tabular} & \begin{tabular}[c]{@{}c@{}}Does the consumer \\ understand how a \\ mechanic works?\end{tabular}                                                      \\ \hline
\begin{tabular}[c]{@{}c@{}}Personal voice \\ assistant \end{tabular}                   & \begin{tabular}[c]{@{}c@{}}A friend's \\ friend \end{tabular}              & \begin{tabular}[c]{@{}c@{}}Facebook post\end{tabular} & \begin{tabular}[c]{@{}c@{}}Has the friend's \\ friend had \\ experience with \\ high-tech?\end{tabular} & \begin{tabular}[c]{@{}c@{}}Is the consumer \\ tech-savvy?\end{tabular}                                                      \\ \hline
\end{tabular}
\end{table}

Whether a consumer can assess the precision of information should also have a significant effect on how a market functions. Consumers who are familiar with the source of information are more likely to recognize its precision. Back to our earlier example, a consumer who reads the post on social media, might have followed the reviewer for a long time. As a result, the consumer knows whether or not the reviewer is a source that generates precise information in the digital camera category. Alternatively, a consumer who does not know the reviewer may coincidentally encounter this post on social media. This consumer might have significantly more difficulty in assessing its precision \citep{Foroughifar20}. Accordingly, because different consumers have different knowledge about the source of information, they might reach different conclusions upon receiving the same information. A recent survey by Podium \citep{Podium17} suggests that 66\% of consumers do not trust information from a source they are unfamiliar with.

Another factor that may help a consumer to recognize the precision of information is her own expertise in the category. Consider a consumer who obtains information through WOM about a new high-tech product (e.g. a personal voice assistant). Here, if the consumer is not familiar with the product category or does not understand the purpose of the new technology, she may not be able to evaluate the precision of the WOM. As a consequence, both high precision and low precision information may have the same effect on her behavior. However, if the consumer is tech-savvy, she may be able to distinguish between precise and imprecise WOM. Accordingly, she may take different actions upon receiving high versus low precision information. The point here is that consumers are often heterogeneous in terms of recognizing precise versus imprecise information (e.g. because they are unfamiliar with the source or with the product category); this difference can have a substantial impact on market behavior.

From a firm's perspective, \textit{information precision} and \textit{consumer ability to assess precision} are two important factors that affect profit and a firm may have an incentive to influence them. First, the firm may have strategies to improve the precision of information; sponsoring a third party to generate precise information or taking action to limit imprecise information. 
Second, a firm might be able to help individual consumers assess information precision. For instance, a firm might promote transparency in "the source" of information or publish guidelines to educate consumers how to interpret third-party information. A firm might even provide "expert" badges to highlight online spokespeople who have expertise in the product category; this would help consumers recognize sources that disseminate precise information. Such strategies might increase the consumer's ability to sort high precision information from information that is less precise.
\footnote{These strategies are not the only ways that firms can influence information in the market. Firms may also "directly" transmit/distort information to influence consumers’ perceptions and persuade them that their product is of high quality. These direct strategies are well-studied in the marketing and economics literature \citep{KamenicaGentzkow11,BergemannMorris19}. Our work approaches this problem from a different angle and abstracts away from the possibility that the firm can directly transmit/distort information. We study novel forces in isolation from direct strategies to clearly understand their market consequences. Moreover, for third-party information transmission, a firm has little control over the message; it may be so that it is more practical for a firm to affect the information through indirect actions.}

Given the strategies that a firm might have to influence the precision of third party information and/or consumers' ability to assess precision, our objective is to determine when a firm should exercise these strategies. The specific research questions we aim to answer are: (1) How do prices and profits change when the precision of third-party information changes? (2) How does the incentive of a firm to impact precision differ when it is high quality versus low quality? (3) Does a firm prefer consumers with high skill in assessing information precision or consumers who have difficulty assessing precision?

We consider a monopoly firm that sells a new product to a unit mass of consumers. The firm knows the quality of its product (high or low) and consumers are uninformed about product quality. Consumers costlessly receive a signal which is imperfectly informative about product quality. Consumers may or may not be able to evaluate the precision of this signal. If the consumer knows the precision of the signal, she distinguishes between signals of high versus low precision (a \textit{sophisticated} consumer). However, if she cannot asses the precision of a signal, she applies the same precision factor to any signal she receives (a \textit{naive} consumer).

Intuitively, one expects that more precise information should benefit a firm when it is high quality and hurts the firm when it is low quality. Our results show that is not always the case. When the firm is low quality, it always suffers from more precise information. However, when it is high quality, it may also suffer from more precise information. Specifically, as precision of information increases, a high quality firm's profit first decreases and then increases. 
The intuition for this is as follows. When the precision of information is low, there is low dispersion in consumer valuations (posteriors). Here, the firm sets a low price and sells to all consumers. As precision increases, there is more dispersion in valuations, meaning that consumers who obtain a bad signal, have a lower valuation, and consumers who obtain a good signal, have a higher valuation. When precision is sufficiently high, broad dispersion causes the firm to switch and only serve consumers with high valuations (by charging a high price). Therefore, as precision increases, the high quality firm's profit first decreases and then increases. 

An important insight from our first result is that the precision range in which a high quality firm gains or suffers from more precise information depends on how skilled consumers are at recognizing precision. The same is not true for a low quality firm. At some point, improvements in information precision and consumer ability to assess precision impair the ability of low quality firm to survive.

A second key finding of our model is that when the firm is high quality, it prefers a market of naive versus sophisticated consumers not only when precision is low, but also when precision is high. This is surprising because it seems that when a firm is high quality, it should prefer a market where consumers can distinguish between precise and imprecise signals. 
However, we find that when precision is in the highest range, the equilibrium price in the naive market is higher: this makes the market more profitable for the firm when it is high quality. Exactly the opposite is the case for when the firm is low quality. When the firm is low quality, it prefers a sophisticated market when precision is in the highest range; despite the lower price, it sells more often.

Our third finding relates to the firm's incentive to help consumers assess precision. We show that the firm's profit often decreases with the consumer ability to recognize precision. Specifically, when the fraction of sophisticated consumers is below a threshold, the firm suffers from improvements in the consumer ability to assess precision. However, when the fraction is above the threshold, the reverse is true. Counter-intuitively, even when information is highly precise, whether or not a firm has incentive to help consumers to assess precision depends on the fraction of market which is sophisticated. This illustrates the importance of consumer ability to recognize precision as a determinant of equilibrium firm behavior. We also provide several extensions which demonstrate the robustness of our model.


The remainder of the paper proceeds as follows. Section \ref{Litreview} reviews the related literature and discusses our contribution. Section \ref{base_model} presents the base model and the main results. In Section \ref{extensions}, we explore extensions to the base model. Finally, in Section \ref{conclusion}, we conclude and discuss the managerial implications of our work.

\section{Related Literature} 
\label{Litreview}

Our study is related to the literature on firm incentives to provide information to consumers. \cite{LewisSappington94} examine how a seller provides information to potential buyers about their taste for its product. They show that the seller provides either high precision information or no information. Similar results are presented by \cite{JohnsonMyatt06} using an aggregate demand function. 
These studies examine a situation in which the seller does not know its quality (there is no signalling through price) and the information precision is common knowledge. We study a different situation where the seller knows its quality and consumers have heterogeneous ability to assess precision. So, our first contribution is to extend this literature to an environment in which both third-party information and firm pricing strategies can signal product quality. This is important because we find that the U-shaped behavior of profit with respect to information precision breaks down when the firm's quality is low. 

Our second contribution is to show that the consumer ability to assess precision plays an important role in the firm's equilibrium strategies: it determines the precision range in which a high quality firm suffers from higher precision information. In addition, we show that it is not always preferable for a high quality firm to sell to consumers who are skilled at assessing precision. In fact, the firm may have incentive to sell in a market where consumers cannot evaluate precision (e.g. low transparency in information sources). Another distinction of our model is that the signals in our model come from third parties. As a consequence, the firm does not have direct control on information precision. This stands in contrast to \cite{LewisSappington94} and \cite{JohnsonMyatt06} who assume the information is provided by the seller and the seller controls the precision.\footnote{This also distinguishes our paper from the literature on information design \citep{BergemannMorris19}.} 

Several other studies examine conditions under which firms benefit from providing information \citep{KuksovLin10, GuXie13, Brancoetal16}. Here, firms choose strategies to make the information they transmit credible \citep{MilgromRoberts86, Shin05, Miklos13}. In addition, third-party information and social interactions have become an important source of information which affect consumer beliefs and decisions \citep{SI05}. Our study contributes to this literature by introducing a new form of consumer heterogeneity: the ability to evaluate the precision of information.\footnote{Other papers have considered marketing questions in the context of third-party reviews \citep{ChenXie05, Mayzlin06, ChenXie08, KuksovXie10, PeiMayzlin19}. This work does not examine the impact that the precision of information has on marketing.} A unique aspect of our study is that the information distribution follows a hierarchical structure. This allows us to model consumer heterogeneity in terms of the ability to distinguish "information precision".

To summarize, it is well established that the precision of information is important in consumer decision making.\footnote{The consumer behavior literature has also studied the importance of information credibility in determining the persuasiveness of communication; see for example \cite{Chiken80}.} Firms may even have an incentive to manipulate information precision. The standard assumption in the literature is that the precision of information is common knowledge and fixed. Yet, \cite{Zhaoetal13} provide empirical evidence that the credibility of third-party information can vary. They show that more credible information has a greater effect on the consumer's purchase decision. \cite{Mayzlin06} speculates that consumers cannot always distinguish between precise and imprecise information.  Building upon these studies, we further argue that there is heterogeneity in consumers' ability to assess information precision. In addition, both the precision and the ability of consumers to distinguish high precision from low precision information varies across markets. Accordingly, our contribution is to examine the impact of two factors on how markets unfold: a) the precision of information from third parties which is often driven by the age of the category and, b) the ability that consumers have to recognize precision. We also examine firm strategies when there is heterogeneity in consumers’ ability to recognize precision.  Today, this is important because evaluating the precision of information in the digital world is challenging. 

The context for our model is a recently introduced (experience) product by a monopolist. Consumers need information to learn about the product and third-party information is important because consumer knowledge is assumed to be low. When a new product is introduced, there is significant variability in the quality of information sources, and consumers may not be able to distinguish high precision sources from low precision sources. As people become more familiar with a category, a lexicon for the product category arises and consumers discuss the category more easily. Hence, it becomes easier to discriminate between high precision and low precision information. In the following section, we present the model that we use to analyse this topic.

\section{The Base Model}
\label{base_model}
The base model entails a seller that offers a new product to a unit mass of consumers. Nature randomly determines the product quality $Q \in \{G,B\}$, where $G$ represents good quality and $B$ represents bad quality. We assume that consumers' reservation price for the good and the bad quality product are exogenously given by $v_G=1$ and $v_B \in [0,1)$. The firm observes the quality but consumers do not. Consumers' prior beliefs are common knowledge, $Pr(G)=\frac{1}{2}$.\footnote{We assume priors are the same across consumers for simplicity. Homogeneous priors is a reasonable assumption for a product that is new; this is commonly assumed in the literature \citep{BergemannValimaki97, Boseetal08, Papanastasiou17}.} The firm sets price to maximize expected profit. Consumers costlessly observe the price and an independent signal about product quality, $\sigma$, update their beliefs about the quality, and then decide to buy either one unit of the product or nothing. The signal is obtained from a third-party that provides information about product quality. One can think of third-parties as friends, online spokespeople, and random posts on social media platforms.\footnote{More generally, the signal can be interpreted as any information not generated by the seller that helps the consumer update her belief regarding product quality. The assumption that the consumer obtains "one" signal has no qualitative bearing on our model insights because all the third-party information that a consumer obtains before making a purchase decision can presumably be summarized in one signal. In addition, we study products that are relatively inexpensive (relative to income). So, it is unlikely that a consumer performs costly search to obtain many signals from different sources.} 

The content of a signal is described by a pair $(q,w)$, where $q \in \{g,b\}$ is the quality that the signal indicates, and $w \in \{h,l\}$ is the precision of the signal. The precision of the signal is related to how informative the signal is and it determines the conditional probability of the signal valence supporting the true product's quality. Given the precision value, $w$, and the true product quality, $Q$, the conditional probabilities of observing signal $q$ are defined as follows:\footnote{This "precision criteria" that is defined based on the property that more informative signals lead to a more disperse distribution of conditional expectations is commonly used in the literature \citep{Blackwell51,LewisSappington94, ChenXie08, MayzlinShin11}. See \cite{GanuzaPenalva10} for a detailed discussion.}
$$Pr(q=g\text{ }|\text{ }Q=G;w)=Pr(q=b\text{ }|\text{ }Q=B;w)=w$$

We further assume that the signal precision $w$ is \textit{ex ante} equally likely to be high or low, i.e., $Pr(w=h)=\frac{1}{2}$, and the signal generating process is common knowledge. For signals to have value, the posterior probability of a state, following a signal which supports that state, needs to be higher than the prior probability.  Thus, the precision is bounded by $[\frac{1}{2},1]$.\footnote{In general, when there are $N$ states, the precision is bounded by $[\frac{1}{N},1]$} A signal of precision $\frac{1}{2}$ is uninformative and a signal of precision 1 is perfectly informative. For simplicity, we assume $0.5= l<h\leq1$ throughout.\footnote{The main findings of our model remain qualitatively similar if we allow both $h$ and $l$ to be flexible.}

Note that the hierarchical relationship between signal precision and the valence of the signal implies that there are two levels of information regarding its content. The first level of information is the precision of the signal which is independent of $Q$. The second level of information is the valence of the signal, which depends on both precision and the true quality of the product. For a given product quality, $Q$, these two levels are depicted below.\footnote{$F_{q|w,Q}(q|w,Q)$ is the conditional c.d.f which depends on both $w$ and $Q$.}
\\
\begin{center}
\begin{tikzpicture}
\matrix[matrix of math nodes, column sep=20pt, row sep=20pt] (mat)
{
    & w & \\ 
    & q &  \\
};

\foreach \column in {2}
{
    \draw[->,>=latex] (mat-1-2) -- (mat-2-\column);
}

\node[anchor=east] at ([xshift =-40pt]mat-1-2) 
{$w \sim \text{uniform
}\{h,l\}$};
\node[anchor=east] at ([xshift =-40pt]mat-2-2) 
{$q \sim \text{F}_{q|w,Q}$} ;

\end{tikzpicture}
\end{center}
A consumer in this model receives one of 4 types of signals:
$$\sigma \in \big\{\sigma_{q,w} \text{    } \big| \text{    } q \in \{g,b\}\text{  ,    }w\in\{h,l\}\big\} \equiv \big\{\sigma_{g,h},\sigma_{g,l},\sigma_{b,h},\sigma_{b,l}\big\}$$ 
In practice, a consumer knows whether a signal is positive or negative. However, people do not always know the precision. They may have to make inferences about it. The high cost of assessing precision, the anonymity of the signal's source, and a lack of consumer expertise in the product category are conditions which create difficulty in accessing signal precision. We assume all consumers observe whether a signal supports good quality or bad quality. But they may or may not be able to distinguish the precision of the signal. A consumer who distinguishes precision is called \textit{sophisticated} and a consumer who does not distinguish precision is called \textit{naive}. Throughout, we assume that consumer types are exogenous. 

\subsection{The Consumer Belief Updating Process}
As noted earlier, consumers are heterogeneous in their ability to distinguish signal precision. This implies that consumers are heterogeneous in terms of their posterior beliefs. Hence, we present consumers' posterior beliefs based on the signals they receive. 

\subsubsection{Sophisticated Consumers} \label{subsec_soph}
Suppose a signal, denoted by $\sigma_{q,w}$, is received by a sophisticated consumer. Since the consumer recognizes the precision, the conditional probabilities of each valence, $Pr(q|Q;w)$, depends on both the product quality, $Q \in \{G,B\}$, and the precision, $w \in \{h,l\}$, as shown in Table \ref{conditional_p_soph}:\\
\begin{table}[ht]
\centering
\caption{Conditional probabilities for sophisticated consumers}
\label{conditional_p_soph}
\begin{tabular}{@{} ccc @{}}
\hline
 & Pr($q|G;w$) & Pr($q|B;w$) \\
\midrule
$q=g$ & $w$ & $1-w$ \\
\midrule
$q=b$ & $1-w$ & $w$ \\
\bottomrule
\end{tabular}
\end{table}\\
A sophisticated consumer uses these conditional probabilities to derive signal-specific conditional probabilities, $Pr(\sigma_{q,w}|Q)=Pr(q|Q;w)Pr(w)$, and updates her belief accordingly. The consumer's posterior belief, upon receiving a signal, follows Bayes' rule:
 \begin{equation}
 Pr(G|\sigma_{q,w}) =\frac{Pr(\sigma_{q,w}|G)Pr(G)}{\sum_{Q}
 {Pr(\sigma_{q,w}|Q)Pr(Q)}}
 \end{equation}
 This leads to the posterior beliefs shown in Table \ref{Posterior_sophisticated}.\\
\begin{table}[ht]
\centering
\caption{Posterior beliefs of sophisticated consumers}
\label{Posterior_sophisticated}
\begin{tabular}{@{} ccc @{}}
\hline
 \ & Pr($G|.$) & Pr($B|.$) \\
\midrule
$\sigma_{g,h}$ & $h$ & $1-h$ \\
\midrule
$\sigma_{g,l}$ & $0.5$ & $0.5$ \\
\midrule
$\sigma_{b,l}$ & $0.5$ & $0.5$ \\
\midrule
$\sigma_{b,h}$ & $1-h$ & $h$ \\
\bottomrule
\end{tabular}
\end{table}
\subsubsection{Naive Consumers}\label{subsec_naive}
Suppose a signal $\sigma_{q,w}$, is received by a naive consumer. Because the consumer is naive, she does not recognize the precision of the signal. Here, we drop the second subscript for ease of exposition and show the signal by $\sigma_q$. The naive consumer forms an expectation about the precision, defined by $\bar{w}=\frac{l+h}{2}=\frac{1/2+h}{2}$, and applies this precision to any signal she obtains. Here, the conditional probability of each valence, $Pr(q|Q)$, is defined as in Table \ref{conditional_p_naive}:\footnote{Note that $Pr(q|Q)=\sum_{w} Pr(q|Q;w)Pr(w)=\sum w Pr(w)=\bar{w}$}
\\
\begin{table}[ht]
\centering
\caption{Conditional probabilities for naive consumers}
\label{conditional_p_naive}
\begin{tabular}{@{} ccc @{}}
\hline
 & Pr($q |G$) & Pr($q |B$) \\
\midrule
$q=g$& $\bar{w}$ & $1-\bar{w}$ \\
\midrule
$q=b$ & $1-\bar{w}$ & $\bar{w}$ \\
\bottomrule
\end{tabular}
\end{table}
\\
In this case, the signal-specific conditional probabilities coincide with the conditional probabilities of the valence, $Pr(\sigma_{q}|Q)=Pr(q|Q) \text{, }\text{ }Q \in \{G,B\}$. Using these probabilities and Bayes rule, the naive consumer's posterior beliefs are as follows:\\
\begin{table}[ht]
\centering
\caption{Posterior beliefs of naive consumer}
\label{Posterior_naive}
\begin{tabular}{@{} ccc @{}}
\hline
\ & Pr($G|.$) & Pr($B|.$) \\
\midrule
$\sigma_g$ & $\frac{1+2h}{4}$ & $\frac{3-2h}{4}$  \\
\midrule
$\sigma_b$ &$\frac{3-2h}{4}$ & $\frac{1+2h}{4}$ \\
\bottomrule
\end{tabular}
\end{table}
\\
As expected, the naive consumer's posterior belief does not vary across signals of different precision.

\subsection{The Game}

The monopolist sells the product at price $p$, to a unit mass of consumers. The firm knows product quality, $Q \in \{G,B\}$, but consumers do not. The marginal cost of production is constant and normalized to zero.\footnote{This implies that all products have positive value relative to cost.} We assume a fraction $\lambda$ of consumers are \textit{sophisticated} and the remainder $1-\lambda$ are \textit{naive}, where $\lambda \in [0,1]$ is common knowledge. Each consumer independently receives a signal about product quality and the firm does not observe individual-specific signals. However, the distribution of signals is common knowledge. This implies that the firm knows the value of $h$ and $l$, and infers the signal-specific conditional probabilities derived earlier. Given its quality, the firm sets price $p$ to solve the following maximization problem:
 \begin{equation}
 \label{p}
 \tilde{p}(Q;h,\lambda) \in \underset{p}{\text{argmax}}\text{  }p \times \Big[\lambda \sum_{\sigma \in A} Pr(\sigma|Q) \ \tilde{d}_S(p,\sigma)+(1-\lambda) \sum_{\sigma \in B} Pr(\sigma|Q) \ \tilde{d}_N(p,\sigma)\Big]
 \end{equation}
where $A = \{\sigma_{b,h},\sigma_{g,h},\sigma_{b,l},\sigma_{g,l}\}$ and $B =\{\sigma_b,\sigma_g\}$. In addition, $\tilde{d}_S(p,\sigma)$ and $\tilde{d}_N(p,\sigma)$ represent the demand of sophisticated and naive consumers, respectively. Note that the summation terms in the right hand side of equation \eqref{p}, correspond to the expected demand, which depends on the distribution of signals (with average precision $\bar{w}=\frac{1+2h}{4}$). An exogenous change in $h$ leads to a change in the average precision, and consequently affects the distribution of signals and equilibrium strategies (e.g. the firm's price and the consumers' purchase decisions).

Because the firm knows the quality of its product, price can also signal product quality. Therefore, both the signal and the price contribute to form the consumer's posterior belief about product quality. Without loss of generality, we assume the consumer first observes a private signal and updates her belief to $Pr_k(G|\sigma)$ according to the analysis provided earlier.\footnote{Here, to avoid confusion and distinguish between the posterior belief of naive and sophisticated consumers, we add a subscript $k \in \{N,S\}$.} She then observes the price, $p$, forms a new belief, $\mu_k(\sigma,p)=Pr_k(G|\sigma,p)$, $k \in \{N,S\}$, and decides whether to buy or not. The consumer is rational in that she maximizes expected utility given her private signal, $\sigma$, and the price:  
  \begin{equation}
   \label{d}
  \tilde{d}_k(p,\sigma) \in \underset{d\in\{0,1\}}{\text{argmax}}\text{ }d\times[\mathbb{E}_k(v|\sigma,p)-p] \ , \ k \in \{N,S\}
  \end{equation}
   This implies that the consumer buys when the expected value, $\mathbb{E}_k(v|\sigma,p)=\mu_k(\sigma,p) v_G+(1-\mu_k(\sigma,p))v_B$, is greater than the price.

\subsection{Equilibrium}
The timing of the game is as follows. Nature determines the quality of the product. The firm observes the quality and sets price. Consumers observe private signals and the price. They then update their beliefs about product quality and decide whether or not to purchase. The solution to this game is based on the Perfect Bayesian Equilibrium (PBE). We examine equilibria in pure strategies. 

\begin{definition}
A (weak) PBE is a triple $(\tilde{p},\tilde{d},\mu)$ such that: 

\textit{1. The firm's strategy is sequentially rational: for $Q \in \{G,B\}$, $\tilde{p}$ solves equation \eqref{p}}

\textit{2. Consumers strategies are sequentially rational: for all $\sigma$ and $p$, $\tilde{d}$ solves equation \eqref{d}}

\textit{3. Consumers beliefs are consistent.}
\end{definition}
In a signalling game, it is necessary to specify the on and off equilibrium beliefs. We elaborate in the next subsection.

\subsubsection{Equilibrium Beliefs and Off-Equilibrium Beliefs}
 
 We first examine the nature of equilibrium beliefs in the model. In signalling games, consistency is a required condition for equilibrium. In a separating equilibrium, the price perfectly reveals the quality of the product. As a result, the consumer's posterior will be independent of her private signal. However, in a pooling equilibrium, the firm sets the same price independent of its quality. In this case, the consumer's posterior is formed as a function of her private signal, $Pr_k(G|\sigma)$, which is shown in Tables \ref{Posterior_sophisticated} and \ref{Posterior_naive}. Therefore, conditional on the signal's precision, a signal contains information about product quality independent of the observed price. 
 
 In order to restrict the off-equilibrium beliefs, we use the Intuitive Criterion \citep{ChoKreps87}. This implies that beliefs cannot ascribe positive probability to a player choosing a strategy that is equilibrium dominated. Thus, whenever a consumer observes an "off equilibrium path" strategy which is beneficial for only one type of firm, she assigns probability 1 to that type. Otherwise, her belief is updated according to her private signal, $\mu_k(\sigma,p)=Pr_k(G|\sigma)$. 
 
\subsection{Results}
In this section, we present the equilibrium outcome of the game.\footnote{As noted earlier, price is a function of $Q$, $h$, and $\lambda$. But we suppress $h$ and $\lambda$ for ease of exposition throughout.} We first identify a lower bound for the equilibrium price in this game.

\begin{lemma}
\label{lem_price_interval}
In all equilibria of this game, prices are in the interval of $[v_B,1]$.
\end{lemma}

Consumers are willing to pay a maximum of $v_B$ when the product quality is low and up to 1 when quality is high. Here, even when quality is low, the firm will charge a price equal to or greater than $v_B$. This implies that the equilibrium price lies in the interval $[v_B,1]$ independent of the signals received by consumers. As a result, all prices less than $v_B$ and greater than 1 are dominated. Next, we outline the parameters of a putative separating equilibrium for this game.

\begin{lemma}
\label{lem_pb}
In any separating equilibrium of this game, it must be the case that $\tilde{p}(B)=v_B$.
\end{lemma}

If there is a separating equilibrium, the following inequalities must hold: $v_B \leq \tilde{p}(B)< \tilde{p}(G) \leq 1$. Further, the maximum possible profit when quality is low would be $v_B$. This obtains because in a separating equilibrium, consumers know the product quality from the observed price: the maximum WTP when quality is low would be $v_B$. This implies that in equilibrium the price must satisfy $\tilde{p}(B)\leq v_B$. Together with Lemma \ref{lem_price_interval}, this implies that $\tilde{p}(B)= v_B$. 

Accordingly, a putative separating equilibrium satisfies $v_B= \tilde{p}(B)< \tilde{p}(G) \leq 1$. This means that all consumers purchase the product in equilibrium. But in this situation, the firm has an incentive to deviate if it is low quality.\footnote{Note that we focus our attention on pure strategy equilibria. When $v_B$ is sufficiently high, there might be a mixed pricing strategy equilibrium with partial separation. We discuss this later.} This reasoning leads to Proposition \ref{no_separating}. The derivation of all results are provided in Appendix A.
\begin{prop}
\label{no_separating}
There is no separating equilibrium in this game, i.e., in any (pure-strategy) PBE under the intuitive criterion, the firm sets the same price independent of its quality.
\end{prop}
The consistency of beliefs ensures that consumer posteriors in a separating equilibrium are $\mu_k(\tilde{p}(G),\sigma) = 1$ and $\mu_k(\tilde{p}(B),\sigma) = 0$, independent of $\sigma$ and $k$. Here, signals do not affect consumers' decisions and prices are perfectly informative. We know from Lemmas \ref{lem_price_interval} and \ref{lem_pb} that a separating equilibrium must satisfy $v_B=\tilde{p}(B) \leq \tilde{p}(G) \leq 1$. However, this cannot be an equilibrium because the firm benefits from deviating to $\tilde{p}(G)$ when it is low quality.

In what follows, we first examine the effect of changes in information precision. This can be modelled as an exogenous change in the average precision of signals across consumers, $\bar{w}=\frac{0.5+h}{2}$. Accordingly, a focus of our analysis is comparative statics with respect to the high precision level, $h$.\footnote{Later, we examine an alternative representation of the improvements in information precision. It is based on an exogenous change in the unconditional probability of a high precision signal, $Pr(w=h)$. There, the insights are qualitatively similar to the base model.} We then examine the impact of consumer ability to assess precision by exogenously changing $\lambda$, the fraction of sophisticated consumers.

\subsubsection{A Market with Mixture of Consumers}

To start, we fix $\lambda$ and examine how the equilibrium profit and price are affected by changes in the precision of the signals. Later, we discuss how the equilibrium changes with $\lambda$. 

\begin{prop}
\label{mixed}
$\exists \text{ } \bar{v} \in [0,1]$ such that for $v_B \leq \bar{v}$, the following holds: there is a threshold $h^*(\lambda)$ such that in equilibrium, \\ $(a)$ If the quality of the firm is high, its profit decreases in precision when $h$ is below the threshold, and increases in precision when $h$ is above the threshold\\
$(b)$ If the quality of the firm is low, its profit decreases in precision for all $h \in [\frac{1}{2},1]$.
\end{prop}

As noted earlier, in a pooling equilibrium, consumers' willingness to pay (WTP) depends only on the signal they receive. Hence, there are 5 groups of consumers in this market.\footnote{In our model, sophisticated consumers who receive a low precision ($g$ or $b$) signal are in the same group because a low precision signal is uninformative. In general, there are six groups of consumers.} Table \ref{table:WTP} shows the five possible WTP levels observed in this market. The WTP depends on the valence of the signal ($g$ or $b$), the precision of the signal ($h$ or $l$), and the receiver of the signal (sophisticated or naive). Sophisticated consumers who obtain the best signal, $\sigma_{g,h}$, have the highest WTP (level 5), and sophisticated consumers who receive the worst signal, $\sigma_{b,h}$, have the lowest WTP (level 1). The WTP of sophisticated consumers who receive a low precision signal is in between (level 3). Intuitively, the WTP of naive consumers who obtain good signals (level 4) is in between the sophisticated consumers who receive high precision good signals and low precision signals. Similarly, the WTP of naive consumers who obtain bad signals (level 2) is in between the sophisticated consumers who receive high precision bad signals and low precision signals. In the table, we also show the relative proportion of consumers in the market that fall into each of the 5 groups. Of course these proportions depend on whether the product is good or bad quality. Hence, the table shows the proportions for the cases when quality is $G$ and when quality is $B$. 
\\
\begin{table}[ht]
\centering
\caption{Willingness-to-pay based on the signal obtained}
\label{table:WTP}
\begin{tabular}{|c|c|c|c|c|c|}
\hline
 \begin{tabular}[c]{@{}c@{}} Consumer \\ Type\end{tabular}      & Sophisticated & Naive   & Sophisticated & Naive   & Sophisticated  \\ \hline 
Signal & $\sigma_{b,h}$ & $\sigma_b$   & $\sigma_{b,l}/\sigma_{g,l}$ & $\sigma_g$               &   $\sigma_{g,h}$      \\ \hline
WTP Value    & $(1-h)+h v_B$ & $(1-\bar{w})+\bar{w} \ v_B$ &  $0.5(1+v_B)$ & $\bar{w}+(1-\bar{w})v_B$               &  $h+(1-h)v_B$  \\ \hline
\begin{tabular}[c]{@{}c@{}}WTP Level \\ (Relative)\end{tabular}  & 1 & 2 & 3 & 4 & 5 \\ \hline
\begin{tabular}[c]{@{}c@{}}Population \\ Mass \\ ($Q=G$)\end{tabular} & $\frac{\lambda}{2}(1-h)$ & $(1-\lambda)\frac{0.5-h}{2}$ & $\frac{\lambda}{2}$ & $(1-\lambda)\frac{0.5+h}{2}$ & $\frac{\lambda }{2}h$ \\ \hline 
\begin{tabular}[c]{@{}c@{}}Population \\ Mass \\ ($Q=B$)\end{tabular} & $\frac{\lambda}{2}h$ & $(1-\lambda)\frac{0.5+h}{2}$ & $\frac{\lambda}{2}$ & $(1-\lambda)\frac{0.5-h}{2}$ & $\frac{\lambda }{2}(1-h)$ \\ \hline
\end{tabular}
\end{table}
\\
The firm sets price to maximize profit. There are five options for the optimal price in this situation, corresponding to each level of WTP in the market. Any price other than these is dominated. In a context where the signals are evenly split between high and low precision, WTP level 5 is not an optimal price because it results in too few sales. Thus, the firm chooses among the remaining four WTPs. Given $\lambda$, the optimal strategy depends on the precision level, $h$. 
\begin{figure}[ht]
\centering
 \caption{The firm's equilibrium strategy ($v_B=0.1$)}
  \label{mix_figure_R4}
    \includegraphics[scale=0.6]{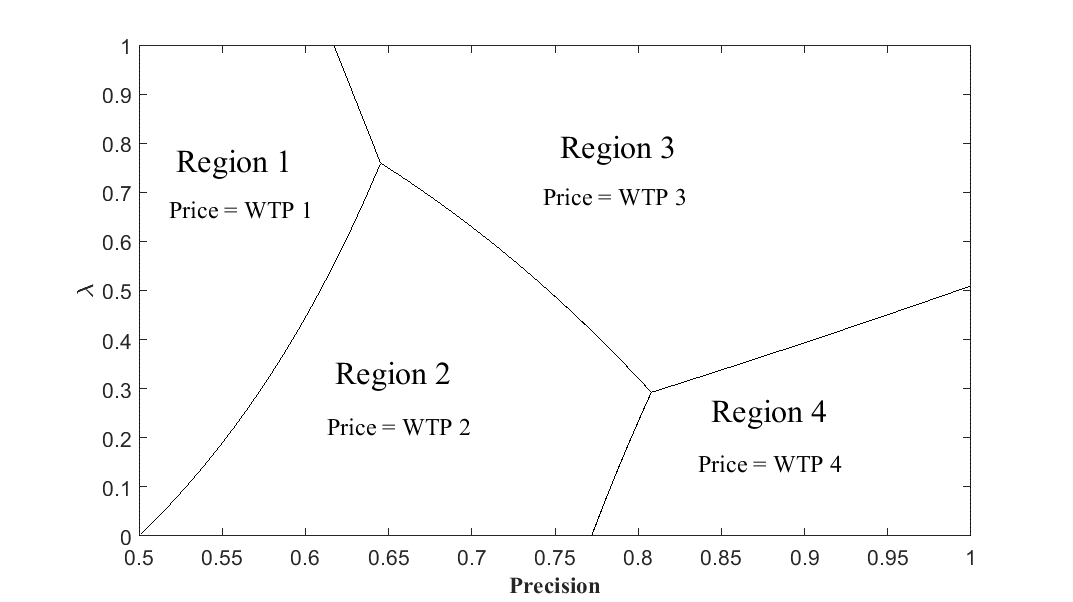}
\end{figure}

First, we assume precision is at the lowest level, $h=\frac{1}{2}$. In this case, the WTP of all consumers coincide at level 3, $\frac{1+v_B}{2}$, and this is the optimal price. Here, all consumers buy independent of their signal.  However, when precision is greater than $0.5$, WTPs spread across five values shown in the second row of Table \ref{table:WTP}. An increase in precision has two effects: 1) relative to WTP level 3, WTP levels 1 and 2 decrease, while WTP levels 4 and 5 increase (the dispersion in consumer posteriors increases).\footnote{We call consumers with WTP levels 1 and 2 as the "low end" of the market and consumers with WTP levels 4 and 5 as the "high end" of the market throughout.} 2) The relative mass of each WTP level changes. This latter effect is critically dependent on the quality level. If the product is high quality, the population mass goes from levels 1 and 2 (low end of the market) toward levels 4 and 5 (high end of the market). However, if the product is low quality, the reverse takes place.

As $h$ increases from $0.5$, the firm can either keep price at $\frac{1+v_B}{2}$ or switch to any of the alternative WTPs. The optimal price in a pooling equilibrium maximizes the firm's profit when it is high quality. So, we first explain the intuition from the firm's perspective when it is high quality. If the firm reduces the price to the lowest WTP, it sells to all consumers. But if it sets a higher price, it loses consumers with lower WTPs. When $h$ is close to $0.5$, it is optimal to keep price as low as possible to not lose consumers. This corresponds to Region 1 in Figure \ref{mix_figure_R4}. Here, there is a full market coverage, but the price decreases with $h$. Hence, the profit is decreasing in $h$.

As $h$ increases further, if the quality is high, the population mass goes from the low end of the market (WTP levels 1,2) to the high end of the market (WTP levels 4,5). This means that losing consumers with WTP levels 1 and 2 becomes less costly. At some point, the (high quality) firm finds it worthwhile to implement a discrete jump in price at the cost of losing the low end of the market. The discrete jump in price depends on $\lambda$. If $\lambda$ is close to 1, naive consumers and the corresponding mass of consumers at WTP level 2 is small, so the firm switches from WTP level 1 directly to WTP level 3. This happens in the top part of Figure \ref{mix_figure_R4} where there is a border between Regions 1 and 3. In Region 3, the low end of the market is not served. So, as $h$ increases, the price remains at WTP level 3 but, given the quality is high, sales increase due to the shift in the population mass from the low end to high end. As a result, the firm's profit increases in precision if it is high quality.\footnote{The constant price is an artifact of the assumption that $l=\frac{1}{2}$. In a model where $l>\frac{1}{2}$, and both $h$ and $l$ increase exogenously over time, two prices are possible (related to the WTP of consumers that receive a low precision signal): one which decreases in $l$ and the other which increases.} If $\lambda$ is smaller, then the price jump is from WTP level 1 to level 2 (Region 1 $\rightarrow$ Region 2). Here, the low end of the market is partially covered, so the price decreases with $h$. In Region 2, an increase in $h$ results in a negative effect (price $\downarrow$) and a positive effect (sales $\uparrow$ due to mass shift) on the high quality firm's profit. But the former effect dominates and the profit continues to decrease with $h$ in Region 2. 

When $h$ is sufficiently large, if the quality is high, eventually the firm switches to only serve the high end of the market, for all $\lambda$ (Regions 3 and 4). As noted earlier, when $\lambda$ is close to 1, the firm switches directly from WTP level 1 to level 3. If $\lambda$ is in a middle range, the firm first switches from WTP level 1 to level 2, then to level 3, and finally to level 4 (Region 1 $\rightarrow$ Region 2 $\rightarrow$ Region 3 $\rightarrow$ Region 4). However, if $\lambda$ is close to zero, the mass of sophisticated consumers (the fraction of the market at WTP level 3) is small, so the firm switches from WTP level 2 directly to level 4 (Region 1 $\rightarrow$ Region 2 $\rightarrow$ Region 4). In Region 4, because the mass of naive consumers is high, the high quality firm finds it beneficial to set the price at WTP level 4 so that it sells at a higher price but at the cost of losing a small fraction of sophisticated consumers who receive low precision signals (WTP level 3). 

Table \ref{table:mixed_eqn} summarizes the analysis of this section and extends it to the case where the monopolist has a low quality product. For the low quality case, the logic for pricing is unchanged because we showed earlier that the equilibrium price is the same independent of the product quality. However, the behavior of sales with respect to precision is reversed because as $h$ increases, the mass goes from the high end of the market to the low end. As a result, in general, sales decrease with $h$. Consequently, when the quality is low, the profit of the firm decreases in $h$ even when $h$ exceeds the threshold $h^*(\lambda)$. 
\begin{table}[ht]
\centering
\caption{What happens when $h$ increases, given $\lambda$}
\label{table:mixed_eqn}
\begin{tabular}{|c|c|c|c|c|c|}
\hline
\multicolumn{2}{|c|}{}          & \textbf{Region 1} & \textbf{Region 2} & \textbf{Region 3} & \textbf{Region 4} \\ \hline
\multirow{3}{*}{\textbf{Q = G}} & \textbf{Price}  &    $\downarrow$      &    $\downarrow$      &      Constant     &     $\uparrow$     \\ \cline{2-6} 
                       & \textbf{Sales}  &      Constant     &     $\uparrow$     &   $\uparrow$       &      $\uparrow$    \\ \cline{2-6} 
                       & \textbf{Profit} &     $\downarrow$     &      $\downarrow$    &      $\uparrow$    &       $\uparrow$   \\ \hline
\multirow{3}{*}{\textbf{Q = B}} & \textbf{Price}  &    $\downarrow$      &     $\downarrow$     &       Constant    &     $\uparrow$     \\ \cline{2-6} 
                       & \textbf{Sales}  &      Constant     &     $\downarrow$     &      $\downarrow$    &      $\downarrow$    \\ \cline{2-6} 
                       & \textbf{Profit} &     $\downarrow$     &      $\downarrow$    &    $\downarrow$      &    $\downarrow$      \\ \hline
\end{tabular}
\end{table}

The impact of changes in precision on a high quality seller's profit, for the complete range of $\lambda$ is shown in Figure \ref{mix_figure}. The shaded area shows the regions in which the high quality firm suffers from higher precision (Regions 1 and 2 in Figure \ref{mix_figure_R4}). The dark area represents a market in the introductory stage: third-party signals (even good ones) are imprecise/unreliable. The firm heavily discounts its product and consumers enjoy a low price. In a mature market, third parties who generate signals gain experience. This improves the average precision of the signals. At low levels of $h$, increases in $h$ lead to lower profits (and prices) independent of product quality. But when $h$ exceeds the threshold $h^*(\lambda)$, only the high quality firm benefits from increases in precision. An important takeaway is that the precision range in which a high quality firm gains or suffers from better information depends on the value of $\lambda$. As a result, the market's ability to assess precision plays an important role in determining equilibrium behavior. 
\\
\begin{figure}[ht]
\centering
\caption{The impact of precision on the firm's profit when quality is high ($v_B=0.1$)}\label{mix_figure}
    \includegraphics[scale=0.45]{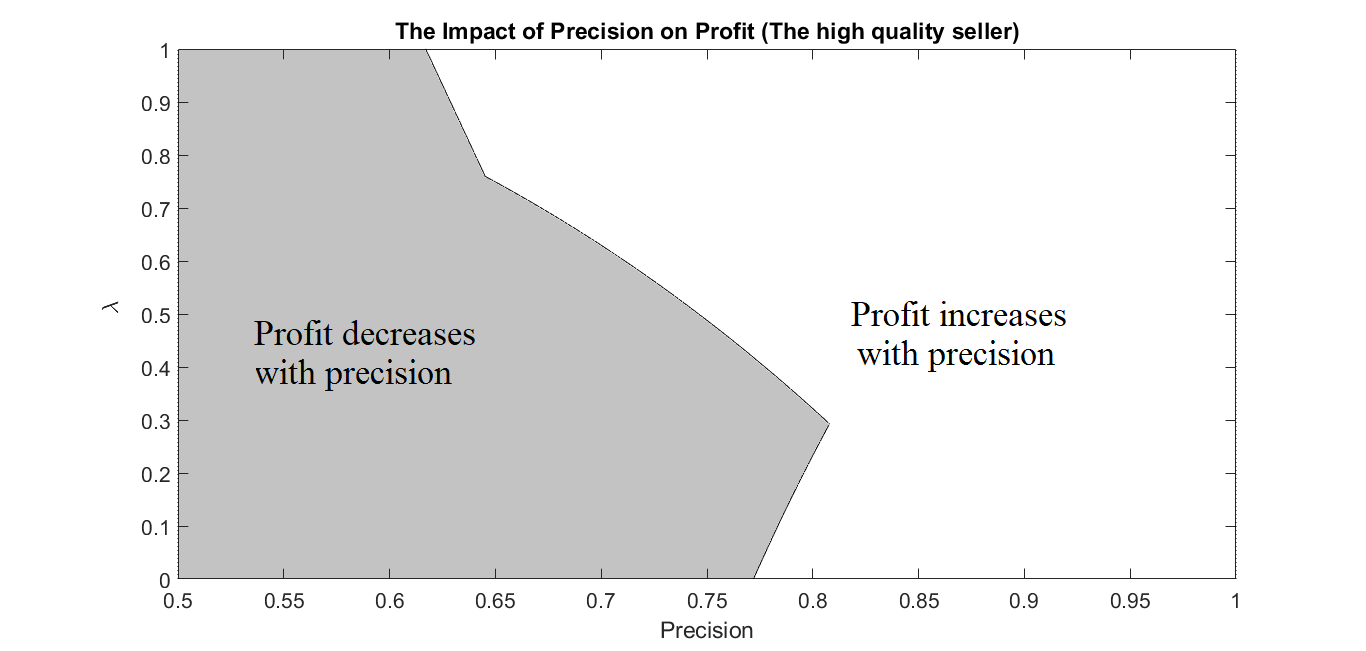}
\end{figure}

Note that Proposition \ref{mixed} holds when $v_B$ is low. If $v_B$ is close to $v_G$ (it exceeds a threshold $\bar{v}$) and the precision is high ($h$ is close to 1), the firm makes few if any sales when it is low quality. Here, the (low quality) firm might defect to $p=v_B$, and the pooling equilibrium unravels.\footnote{Defection is beneficial because at $v_B$ the firm sells to all consumers, and earns profit $v_B\times 1$ which is greater than $\tilde{p} \times E\big[\tilde{d}(\tilde{p},\sigma)\big]$.} In this situation, an equilibrium in pure strategies does not exist. In an extension (Section \ref{extensions}), we show that when $v_B>\bar{v}$, a partially separating equilibrium often exists in which the firms choose mixed pricing strategies. In this situation, the insights are qualitatively similar to the findings when $v_B<\bar{v}$.

We now fix the high precision value $h$ and examine how the equilibrium is affected by changes in $\lambda$, the fraction of market which is sophisticated. To start, we compare two extreme cases where the market is fully naive or sophisticated, i.e. $\lambda$ is 0 or 1. We then turn to a general situation where both types of consumers are present, $\lambda \in (0,1)$.

\subsubsection*{Comparison between Sophisticated and Naive Markets:}
The difference between a sophisticated market ($\lambda=1$) and a naive market ($\lambda=0$) is important because a firm often has to choose between markets and the sophistication of the market might be a key difference. We show that choosing between a sophisticated market and a naive market is not straightforward.

First, we compare the precision threshold across sophisticated and naive markets.

\begin{corollary}
\label{th_n>th_i}
The threshold precision is lower when consumers are sophisticated than when they are naive, i.e., $h^*(1)<h^*(0)$.
\end{corollary}

The precision threshold is a critical value above which signals may deter some consumers from buying. When precision is below the threshold, the market is fully covered and all consumers purchase. However, when the average precision exceeds a threshold, signals deter a significant fraction of consumers from buying. Specifically, when consumers are naive (sophisticated), the threshold determines the precision value above which consumers who obtain a bad signal (a precise bad signal) do not buy. Comparing the thresholds of $\lambda = 0$ and $\lambda = 1$, Corollary \ref{th_n>th_i} shows that the firm serves the whole market for a greater range of precision when consumers are naive versus sophisticated. This obtains because posterior beliefs are less dispersed in a naive market, and the low end of the market is larger and willing to pay a higher price than in a sophisticated market. That is, naive consumers who receive a precise bad signal are more profitable than corresponding sophisticated consumers. As a result, the firm serves \textit{all} consumers when they are naive versus sophisticated for a wider range of precision values.

Next, we examine how the equilibrium profits differ across sophisticated and naive markets. One might expect a firm to prefer a sophisticated market when it is high quality. But our result demonstrates that this is not always the case. 

\begin{corollary}
\label{prefer_firm}
\textbf{\\ \indent a)} If the firm is high quality, it prefers a sophisticated market to a naive market for intermediate values of $h$, while it prefers a naive market to a sophisticated market for low and high levels of $h$.\\
\indent \textbf{b)} If the firm is low quality, it prefers a sophisticated market to a naive market for high values of $h$, while it prefers a naive market to a sophisticated market for low levels of $h$.
\end{corollary}

Figure \ref{fig:comparison} is an illustration of Corollary \ref{prefer_firm}. Note that the relative values of $\underline{h}$ and $\bar{h}$ are as follows: $0.5 \leq h^*(1) \leq \underline{h} \leq h^*(0) \leq \bar h \leq 1$.\footnote{The expressions for $\underline{h}$ and $\bar{h}$ are derived in the Appendix.} We first examine the left panel which demonstrates the firm's profit across sophisticated and naive markets for a high quality seller. To describe the intuition, we compare prices and sales across sophisticated and naive markets. When $h$ is low $h \in (\frac{1}{2},\underline{h})$ the firm prefers to market to naive consumers versus sophisticated ones. Here, the firm sells to all consumers, independent of consumer type. Because the equilibrium price is higher in a naive market than in a sophisticated market, the firm earns more in a naive market. 

When $h$ exceeds $h^*(1)$, the firm's profit is positively related to precision in the sophisticated market. Despite the inflection in the profit curve, as shown in the left panel of Figure \ref{fig:comparison}, the naive market remains more profitable until $\underline h$. In the interval $(h^*(1),\underline{h})$, the precision is low, but prices are lower in a naive market than a sophisticated market.\footnote{This is because naive market serves all consumers but sophisticated market has already switched to a price where some consumers do not buy.} This implies the quantity of sales is higher in a naive market. The higher sales more than make up the lower prices in a naive market, and consequently the firm is more profitable in a naive market.

 %
\begin{figure}[ht]
\centering
\caption{The comparison of profits across naive and sophisticated markets ($v_B=0.1$)}
  \label{fig:comparison}
    \includegraphics[scale=0.55]{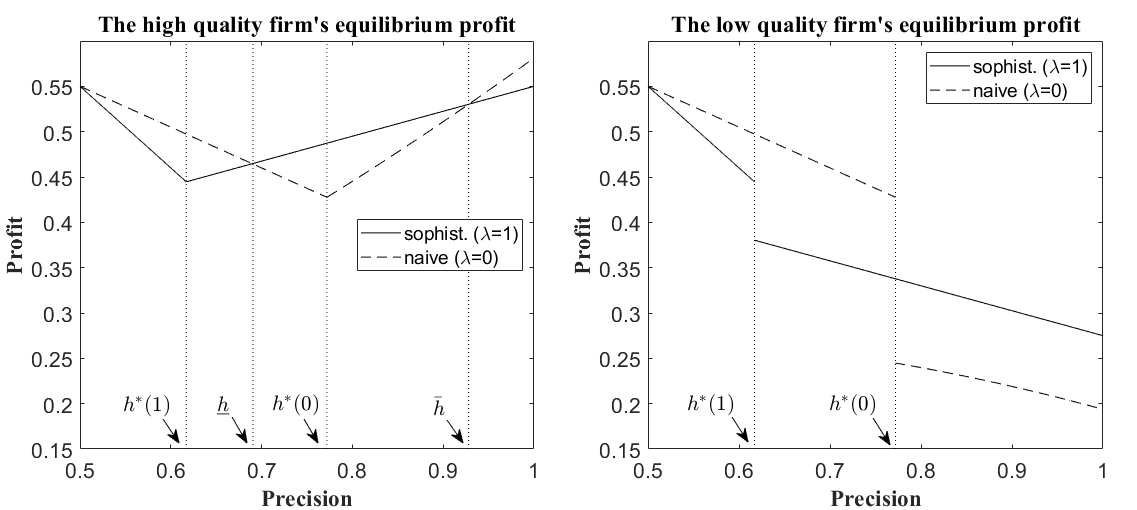}
\end{figure}

When $h$ exceeds $\underline{h}$, the difference between the high quality firm's curve in the sophisticated and in the naive market increases until $h^*(0)$. In the interval $(\underline h, h^*(0))$, the quantity of sales for the firm is smaller in a sophisticated market than in a naive market but the equilibrium price is considerably higher in the sophisticated market. Here, the price effect dominates the sales effect, and the firm prefers sophisticated consumers to naive consumers. 

When $h$ is greater than $h^*(0)$ in the left panel of Figure \ref{fig:comparison}, the firm's profit in the naive market increases with precision. Here, the slope of the profit is steeper in the naive market so that at $\bar h$ the profit curve for the naive market crosses the sophisticated market. When $h \in (h^*(0),\bar h)$, the equilibrium price is lower in a sophisticated market compared to a naive market, but sales are higher. The sales effect is stronger than the price effect so the conclusion is analogous to the previous region: the firm prefers a sophisticated market.

The truly surprising zone of the left panel in Figure \ref{fig:comparison} is $(\bar{h},1)$ where the high quality firm's profit curve for in naive market is higher than in the sophisticated market. Here, the precision of signals is very high. As a result, in a sophisticated market, many consumers obtain high precision good signals. But in a sophisticated market, the price is governed by the expected value of a product with a low precision signal. In contrast, the price in a naive market reflects the value of product with a good signal. The higher prices associated with a naive market more than counterbalance the lower level of sales. Hence, the firm is more profitable in a naive market. This result has an important implication. It means that when high precision signals are very informative, the high quality firm realizes a benefit due to the inability of naive consumers to recognize precision.

Extending this analysis to the case in which the firm is low quality, the results are different. The right panel in Figure \ref{fig:comparison} illustrates this result for the firm is low quality. In this figure, there are two regions of interest. The first region is where $h \in (\frac{1}{2},h^*(0))$. In this region, the firm's profit is higher in a naive market than in a sophisticated one. The intuition is as follows. When $h \in (\frac{1}{2},h^*(1))$ the difference in sales across naive and sophisticated market is zero. But prices are higher in a naive market. As a result, the firm prefers a naive market. When $h^*(1)<h < h^*(0)$, there is a drop in the profit for the sophisticated market due to the discrete jump in price. Here, the price is lower in a naive market, but higher sales more than make up for the lower prices. So, the conclusion is analogous: the firm earns more in a naive market.

The second region in the right panel of Figure \ref{fig:comparison} is $(h^*(0),1)$. Here, the precision is high, but interestingly the firm earns more in a sophisticated market. At the point $h^*(0)$, the price jumps in the naive market so that the firm sells less often in a naive market than in a sophisticated market.\footnote{In this region, sophisticated consumers are not deterred by uninformative bad signals. The firm sells to all consumers who receive low precision signals so the low quality firm sells more in a sophisticated market.} Despite the lower price in the sophisticated market, higher sales make the sophisticated market more profitable for the firm when it is low quality. We can interpret these results as follows. 

In the early stages of a category, the average precision of information is low. Here, a seller (either low or high quality) prefers a naive market to a sophisticated market. This is because the average WTP is higher in a naive market. In more mature categories, the average precision is higher. When precision is in an intermediate range, a high quality firm is better off in a sophisticated market. When the average precision is high, the firm prefers a sophisticated market if it is low quality and a naive market if it is high quality. This obtains because the low quality product sells more in a sophisticated market, but the high quality product is sold at a higher price in a naive market.  

\subsubsection*{Does an Increase in $\lambda$ Lead to Higher Profit?}

 As noted earlier, sophisticated consumers can assess the precision of signals. The ability to assess precision can sometimes be improved through education related to the product or by providing better information about the source of the signal (e.g. more transparency in third-parties who generate information). We define "\textbf{educating}" as an action which helps consumers to better assess the precision of signals (an increases in the fraction of sophisticated consumers). In other words, education is assumed to help the receivers of signals recognize their precision; it has no effect on the exogenous signal generating process. 


The following proposition implies that educating consumers is not beneficial to a firm unless the fraction of sophisticated consumers in the market exceeds a threshold.

\begin{prop}
\label{lambda}
\textbf{\\ \indent a)} Given $h < h^*(1)$, there exists a threshold $\bar{\lambda}(h) \in [0,1]$ such that in equilibrium, the firm's profit is decreasing in $\lambda$ for $\lambda \leq \bar{\lambda}$, and it is independent of $\lambda$ for $\lambda>\bar{\lambda}$
\textbf{\\ \indent b)} Given $h \geq h^*(1)$, there exists a threshold $\bar{\lambda}(h) \in [0,1]$ such that in equilibrium, the firm's profit is decreasing in $\lambda$ for $\lambda \leq \bar{\lambda}$, and it is increasing in $\lambda$ for $\lambda>\bar{\lambda}$.
\end{prop}
Earlier, we explained that the equilibrium price equals one of the WTPs in Table \ref{table:WTP}. As a result, the price is not affected by $\lambda$. Changes in $\lambda$ affect the profit through sales. Figure \ref{mix_figure_lambda} illustrates the result of Proposition \ref{lambda}. The firm's profit (independent of product quality) is decreasing with $\lambda$ in Region 1 (the dark area), independent of $\lambda$ in Region 2, and increasing with $\lambda$ in Region 3.  

In Region 1, the equilibrium price is primarily based on demand from naive consumers: when the precision is low, the equilibrium price is equal to the WTP of naive consumers who obtain a bad signal (WTP level 2). When the precision is high, the equilibrium price is equal to the WTP of naive consumers who obtain a good signal (WTP level 4). As noted earlier, increasing $\lambda$ does not affect the price. However, it affects the fraction of buyers because it shifts mass from WTP groups who buy to WTP groups who do not buy. This implies that sales and the firm's profit decrease with $\lambda$. Hence, there is no incentive for the firm to educate consumers in Region 1. In fact, it is the opposite; the firm would like to make assessing precision more difficult. Counter-intuitively, even when information is very precise ($h$ is close to 1) and the firm is high quality, it prefers a lower fraction of sophisticated consumers.

In Region 2, the equilibrium price is equal to the lowest WTP in the market, WTP level 1. In this situation, all consumers buy and changes in $\lambda$ do not affect the equilibrium profit. As a result, the firm's profit does not depend on $\lambda$. 
 %
\begin{figure}[ht]
\centering
\caption{The impact of consumer sophistication level on the firm's profit ($v_B=0.1$).}
  \label{mix_figure_lambda}
    \includegraphics[scale=0.5]{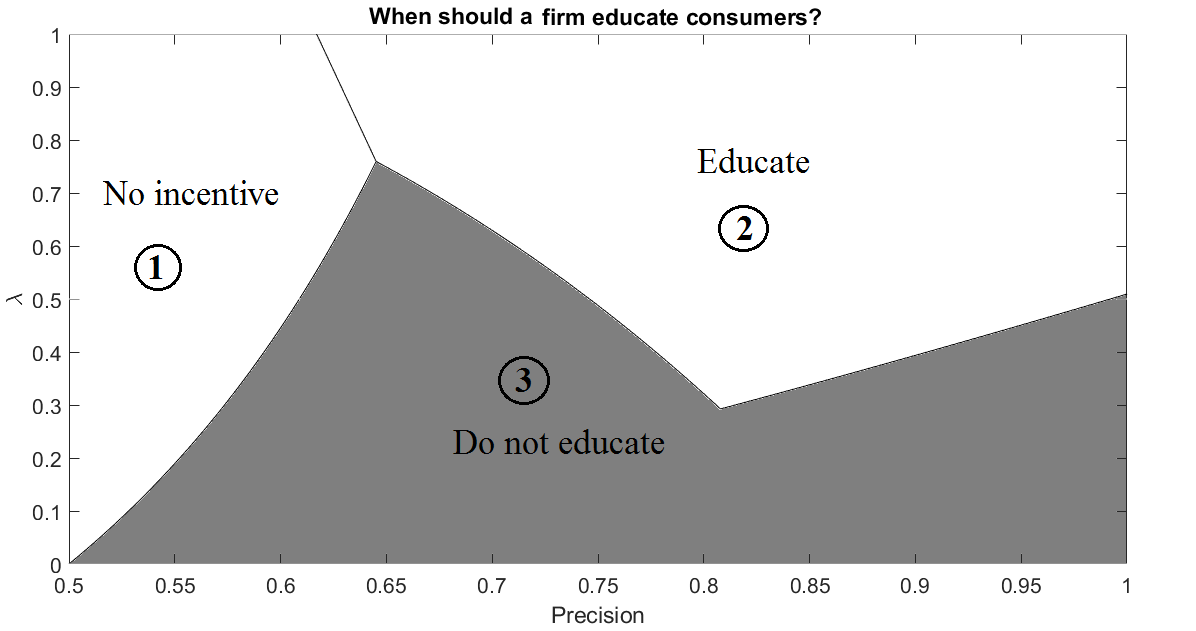}
\end{figure}
%

In Region 3, the equilibrium price equals the WTP of group 3 which corresponds to sophisticated consumers who receive a low precision signal. Here, increasing $\lambda$ increases the fraction of buyers (shifts the population mass from non-buyers to buyers). Hence, educating consumers, increases the fraction of sophisticated consumers and this increases the firm's profit, independent of quality.

\section{Extensions}
\label{extensions}
In this section, we assess the robustness of our results. To simplify, we assume all consumers are the same type and present insights for the case where all consumers are naive, $\lambda=0$.\footnote{Results are qualitatively similar if both types of consumers are present in the market.} First, we analyse a situation where the difference in consumer reservation price for the high and low quality products is small. As noted earlier, in such a situation, the pure strategy equilibrium breaks down. Later, we demonstrate the robustness of our insights to changes in the unconditional probability of a high precision signal, $Pr(w=h)$, and the priors.

\subsection{Partial Separating Equilibrium with Mixed Pricing Strategies}

 The focus of the main model is situations where there is a substantial difference between $v_B$ and $v_G$. Here, we examine market outcomes when the value of the low quality product is close to the high quality product. In such cases, we demonstrate the existence of an equilibrium in mixed strategies.

When consumers are naive, we have shown that when $v_B$ is small, $v_B \in [0,\bar{v}]$, there is a pure-strategy pooling equilibrium (Proposition \ref{mixed}). It is also straightforward to show that when $v_B$ is close to 1, $v_B \in [\bar{v}',1)$, Proposition \ref{mixed} holds.\footnote{Here, it can be shown that $\bar{v}'=\frac{5}{9}>\bar{v}$.} However, when $v_B$ is in an intermediate range, $v_B \in (\bar{v},\bar{v}')$, there are situations with no pure-strategy equilibrium. This occurs when the precision of information, $h$, is close to 1. The shaded area in Figure \ref{mixed-strat-naive} illustrates parameter values for which there is no pure-strategy equilibrium.\footnote{When $v_B$ is in an intermediate range, and $h$ is small, a (pure-strategy) pooling equilibrium as in Propositions \ref{mixed} exists. However, when $h$ is large, there is no pure-strategy equilibrium, either pooling or separating.} Here, a low quality firm has an incentive to deviate from the pooling strategy to $p = v_B$. Yet, separation cannot be an equilibrium according to Proposition \ref{no_separating}. So, the firm might adopt a mixed strategy in equilibrium. We provide an example of such an equilibrium in Proposition \ref{prop:mixed_naive}.

\begin{figure}
\centering
      \caption{Equilibrium type when consumers are naive}
  \label{mixed-strat-naive}
    \includegraphics[scale=0.5]{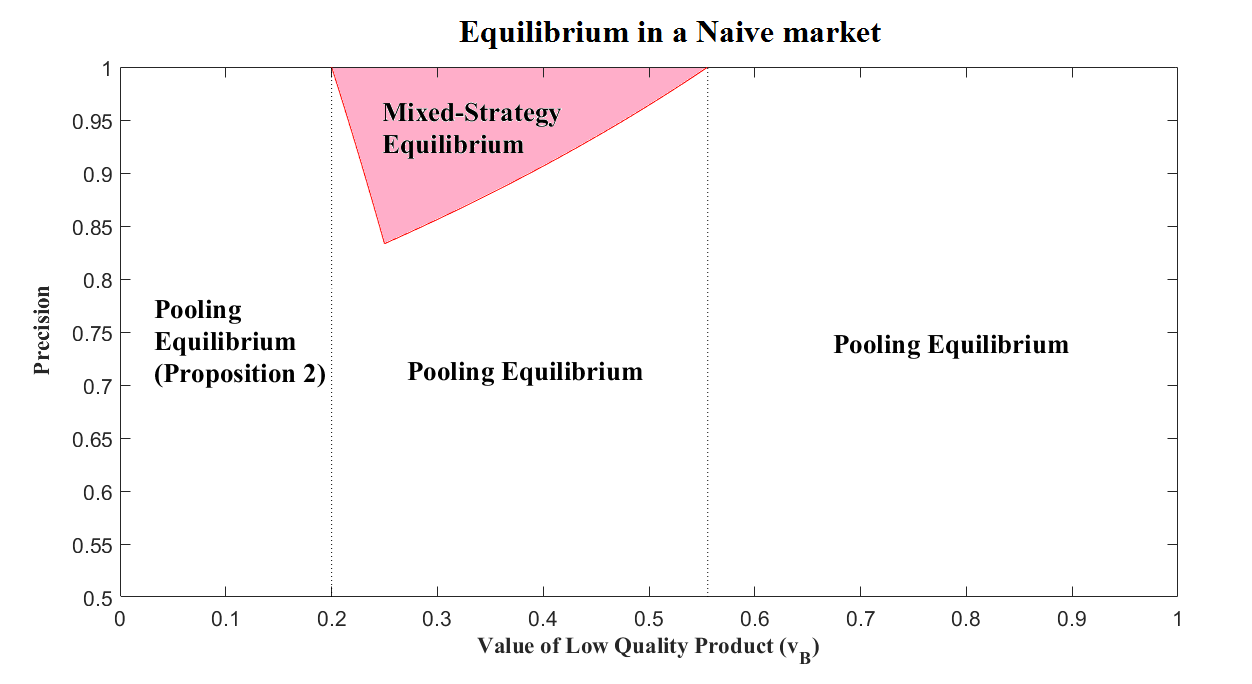}
\end{figure}

\begin{prop}
\label{prop:mixed_naive}
Suppose all consumers are naive. Given $h$, if $v_B \in (\frac{3-2h}{7-2h},\frac{3-2h}{4})$, there is a mixed-strategy equilibrium in which the firm charges a high price when it is high quality, and it mixes between a low price and a high price when it is low quality.
\end{prop}

We identify a mixed-strategy equilibrium in which the low quality firm mixes between a low price ($v_B$) and a high price ($\bar{p}$), and the high quality firm employs a pure strategy and charges a high price ($\bar{p}$).\footnote{Here, the high price reflects the WTP of consumers who get a good signal, $\sigma_{g}$. This is consistent with the restriction that players do not play an equilibrium dominated strategy.} This implies that a low price is perfectly informative, in the sense that consumers are certain that the product is low quality after observing a low price. However, consumers' uncertainty is not resolved after observing a high price, and they form beliefs consistent with the equilibrium mixing strategies; the high price, $\bar{p}$, is an endogenous value which is determined based on the mixing probabilities. 

Without loss of generality, we assume the consumer first obtains a private signal, $\sigma$, and then observes the price. Denote the consumer prior after observing the signal by $\mu(\sigma)=Pr(G|\sigma)$, $\sigma \in \{\sigma_{g},\sigma_{b}\}$, and after observing the price by $\tilde{\mu}(\sigma,p)$. The timing of the game is as follows: first, a firm learns its quality. If it is high quality, it charges a high price, $\bar{p}$. If it is low quality, it mixes between a high price $\bar{p}$ and a low price $v_B$, i.e., it sets $\bar{p}$ with probability $\alpha$ and $v_B$ with probability $1-\alpha$. Consumers observe a private signal and the price, update their beliefs to $\tilde{\mu}(\sigma,p)$, and decide whether or not to purchase the product. Here, when a consumer observes a low price, she is certain that the product is low quality and purchases the product. However, when she observes a high price, she does not know whether the product is high or low quality, because the low quality firm sometimes sets a high price.\footnote{In contrast to Tables \ref{Posterior_sophisticated} and \ref{Posterior_naive}, here the price is an important factor in determining consumers' posterior beliefs. In a pooling equilibrium, the posterior is only a function of the signal.} Accordingly, she updates her prior based on the equilibrium mixing probabilities. So, consistency of beliefs implies that,
\begin{equation}
\label{eqm_beliefs}
    \tilde{\mu}(\sigma,\bar{p}) = Pr(G|\bar{p},\sigma) = \frac{Pr(\bar{p}|G,\sigma)Pr(G|\sigma)}{Pr(\bar{p}|\sigma)} = \frac{\mu(\sigma)}{\mu(\sigma)+(1-\mu(\sigma))\alpha} \text{ } \text{ , } \text{ } \tilde{\mu}(\sigma,v_B) = 0
\end{equation}

When a consumer observes a high price, she only buys if she obtains a good signal. Hence, the maximum WTP in the market is $\bar{p}=\tilde{\mu}(\sigma_g,\bar{p})v_G+(1-\tilde{\mu}(\sigma_g,\bar{p}))v_B$. The high quality firm plays $\bar{p}$ in equilibrium because deviating to other prices is not beneficial. The low quality firm, however, mixes between $v_B$ and $\bar{p}$. Note that when the low quality firm sets a high price, it trades off the windfall of high price against the probability of not selling because when the firm is low quality, it is significantly more likely that a consumer receives a bad signal that deters her from buying. As a result, the low quality firm is indifferent between the low and the high price in equilibrium.
\begin{equation}
\label{indifference}
    \frac{3-2h}{4} \bar{p} = v_B 
\end{equation}
Solving for $\alpha$ leads to $\alpha^*=\frac{1}{v_B}-\frac{4}{3-2h}$. Substituting this result into the profit functions, we find that the high quality firm's profit increases in $h$, which is consistent with the results of Section \ref{base_model}. However, the low quality firm's profit is independent of $h$: the mixed strategy leads to profit of $v_B$ for the low quality firm.

\subsection{Changes in the Unconditional Probability of High Precision}

In the base model, we model improvements in information precision as an increase in the value of high precision, $h$. Alternatively, an improvement in precision might be driven by an increase in the unconditional probability of receiving a high-precision signal. In the base model, the probability is fixed at $\frac{1}{2}$. Here, we relax this assumption and investigate the effect of an exogenous increase in this unconditional probability, holding the value of $h$ fixed.

Consider the case where $v_B=0$ and $Pr(w=h)=1-Pr(w=l)=\gamma$, where $\gamma \in (0,1)$.\footnote{The insights are qualitatively similar for $v_B > 0$, as long as $v_B$ is not too large.} In terms of consumer beliefs, this is similar Section \ref{base_model}, except that the average precision, $\bar{w}$, changes from $\frac{h+l}{2}$ to $\gamma h +(1-\gamma) l$. So, Tables \ref{conditional_p_soph}-\ref{conditional_p_naive} apply but the naive consumers' posterior beliefs are provided in Table \ref{Posterior_naive_new}:\footnote{Note that the conditional probabilities of each valence remain the same as in Tables \ref{conditional_p_soph} and \ref{conditional_p_naive}. However, the signal-specific conditional probabilities, $Pr(\sigma_{q,w}|Q)$, change. Conversely, the posteriors of sophisticated consumers do not change because $Pr(w)$ is in both the numerator and the denominator.}
\begin{table}[ht]
\centering
\caption{Posterior beliefs of naive consumer}
\label{Posterior_naive_new}
\begin{tabular}{@{} ccc @{}}
\hline
\ & Pr($G|.$) & Pr($B|.$) \\
\midrule
$\sigma_g$ & ${{1+\gamma(2h-1)} \over 2}$ & ${{1-\gamma(2h-1)} \over 2}$  \\
\midrule
$\sigma_b$ &${{1-\gamma(2h-1)} \over 2}$ & ${{1+\gamma(2h-1)} \over 2}$ \\
\bottomrule
\end{tabular}
\end{table}
\\
The fraction of high and low precision signals is not fixed at 50/50 in this extension. Here, a $\gamma$ fraction of consumers obtain high precision signals and the remaining $1-\gamma$ obtain low precision signals. When $\gamma$ changes, the signal-specific conditional probabilities, $Pr(\sigma_{q,w}|Q;w)$, and the consequent expected demand functions in \eqref{p} change. In the following, we present comparative statics with respect to $\gamma$.

\begin{prop}
\label{prop:unconditional-naive}
Suppose all consumers are naive. Given $h \in (0.5,1)$, the following hold in equilibrium:\\ 
$(a)$ When the firms is low quality, its profit decreases in $\gamma$ for all $\gamma \in (0,1)$. \\ 
$(b)$ When the firms is high quality, if $h$ is small, the profit decreases in $\gamma$, for all $\gamma \in (0,1)$. \\
$(c)$ When the firms is high quality, If $h$ is in medium range, the profit first decreases and then increases in $\gamma$.\\
$(d)$ When the firms is high quality, If $h$ is large, the profit increases in $\gamma$, for all $\gamma \in (0,1)$.

\end{prop}

The intuition behind Proposition \ref{prop:unconditional-naive} is as follows. From the perspective of naive consumers, there are two types of signals, good signals and bad signals. A firm needs to choose between two strategies: either set a low price and sell to all consumers, or set a high price and sell only to consumers who obtain a good signal. A change in $\gamma$ impacts both the mass of consumers who obtain each of these signals and their WTPs.  

When the quality of the firm is low, an increase in $\gamma$ has two effects: 1) it shifts mass from high WTP to low WTP (increases the mass of people who obtain a bad signal and decreases the mass of people who obtain a good signal), and 2) it increases the dispersion in WTPs. Thus, the result is similar to the base model, the firm's profit decreases in $\gamma$.

However, when the quality of the firm is high, the behavior of profit in $\gamma$ depends on the value of $h$. When $h$ is small, the dispersion in consumer WTPs is so low that, independent of the value of $\gamma$, the firm sets a low price, $p = 1-\bar{w}={{1-\gamma(2h-1)} \over 2}$, and sells to all consumers. That is, the firm's profit is ${{1-\gamma(2h-1)} \over 2}$ for all $\gamma \in (0,1)$. This implies that when quality is high, the firm's profit decreases in $\gamma$ when $h$ is small.

When $h$ is in a medium range, and the firm's quality is high, $\gamma$ has a significant effect on the firm's strategy. When $\gamma$ is low, the firm serves everyone by setting a low price because the dispersion in consumer valuations is narrow. However, when $\gamma$ is high, the firm switches to a high price and sells only to consumers who obtain a good signal because the dispersion in consumer valuations is large. Thus, the firm's profit first decreases and then increases in $\gamma$.

Finally, when $h$ is large and the firm's quality is high, the dispersion in consumers' valuations is so high that, independent of the value of $\gamma$, the firm sets a high price and only sells to consumers with good signals. In this case, the firm's profit increases in $\gamma$ for all $\gamma \in (0,1)$.

\subsection{The Role of the Consumer Prior}
Here, we examine how different priors affect the results of the base model. In Section \ref{base_model}, the consumer's prior, $Pr(G)$, is normalized to 0.5. Here, we relax that assumption by allowing the prior to be flexible between 0 and 1, $0\leq Pr(G)=\mu_0\leq 1$. The consumer's posterior in this situation is derived as, $$Pr(G|\sigma) = \frac{\mu_0 Pr(\sigma|G)}{\mu_0 Pr(\sigma|G)+(1-\mu_0)(1-Pr(\sigma|G))}$$
The expected demand and the firm's profit function can be derived straightforwardly. The equilibrium described in Proposition \ref{prior} is qualitatively similar to that of Section \ref{base_model}.

\begin{prop}
\label{prior}
Suppose all consumers are naive. As long as $v_B$ is not too large, and $\mu_0$ is not too small, there is a threshold $h^*(\mu_0) \in (\frac{1}{2},1)$, such that in equilibrium:\\
(a) If the firm is high quality, its profit decreases in precision when $h<h^*$, and increases in precision when $h>h^*$\\
(b) If the firm is low quality, its profit is decreasing in precision for all $h \in [0.5,1]$
\end{prop}

This proposition implies that the main result of the base model is robust to a change in the prior of the product being high quality. When the prior changes, consumers' WTPs change in the same direction. Therefore, the equilibrium price and profit change accordingly. A higher prior is beneficial to the firm regardless of whether it is high or low quality, because it enhances the posteriors which in turn increases the expected value of the product to consumers. However, it does not change the qualitative behavior of the firm's equilibrium profits with respect to information precision.

\section{Conclusions and Implications}
\label{conclusion}
Recent developments in information technology have significantly influenced information transmission in markets. There are many ways to obtain information about product quality. Information is easily transmitted to consumers from third-parties such as friends. In addition, online spokespeople have become important sources of information about product quality. At the same time, people differ in their ability to assess the reliability of information. As a result, consumers might reach different conclusions upon receiving the same information. 
 
We use a monopoly model to investigate how firm strategy and incentives change with the precision of information and consumer heterogeneity in assessing precision. Our first result shows that, given a combination of sophisticated and naive consumers in the market, when the firm is high quality, as the (average) precision of information increases, the profit first decreases and then increases. However, when the firm's quality is low, the profit always decreases in precision. Importantly, the precision range in which the high quality firm suffers from more precise information depends on how skilled consumers are at recognizing precision. Information transmitted among consumers is relatively inaccurate in the early stages of a category. Thus, when a category is in its introductory stage, information precision is low. Here, the incentives of a firm do not favor improving information precision. In fact, it might be the opposite. Perversely, a firm might gain by spamming imprecise information online. From a welfare perspective, there might be a need for policy intervention to prevent the manipulation of information by firms. In more mature markets, information precision is higher. We find that when the information precision is high, the firm's incentives favor enhancing information precision when product quality is high. Improving precision can be done, for instance, by providing free samples and targeting experts.

When a firm introduces an innovative product, it should target markets which provide the greatest profit. Consumers across markets are often different in their ability to process information. Consumers in one market might be good at distinguishing the precision of third-party information, but consumers in another might not. Our second result shows that if the firm is high quality, it prefers a sophisticated market to a naive market when precision is in an intermediate range. However, it is not always beneficial for a high quality firm to sell its product to a market where consumers are skilled at assessing information precision (e.g. there is transparency in information sources). Surprisingly, when good signals are highly precise and the quality of the firm is high, it prefers a naive market to a sophisticated market. The high price in a naive market more than makes up for lower sales. We also find that when good signals are highly precise and the firm's quality is low, it prefers a sophisticated market. The simple reason is that it does not lose sales because of low precision signals.
 
Our third result concerns the firm's incentive to improve consumers' ability to assess precision. When a category is in its introductory stages, consumers are unsophisticated. The market is new and people have limited ability to communicate about product quality. Hence, the fraction of sophisticated consumers (who can assess precision) is often small. In this case, we find that a firm has no incentive to increase the fraction of sophisticated consumers, even if it is high quality. In mature categories, people are more familiar with the category and the terms used to describe it. This is particularly true of electronics and appliances. As an example, it is difficult to think of discussing GBs of memory and pixels for TV sharpness 25 years ago. The public’s ability to process information is higher in a mature category (the fraction of sophisticated consumers is high). Here, the firm's incentives favor improving consumers' ability to assess precision. The firm can do this in two ways. First, it can facilitate the identification of high precision signals, for instance, by incentivizing third-parties who provide precise information and raising transparency in information sources. Second, the firm can educate consumers by sponsoring training seminars or by facilitating hands-on experimentation. When consumers are better informed about a product, they can better assess the precision of signals they encounter.

We show that our results are robust to situations where the gap between high quality and low quality is small. Here, when the precision of signals is high, there is no pure-strategy equilibrium and the equilibrium entails the low quality firm choosing a mixed pricing strategy. Even here, the results are qualitatively similar to the main model. The notable difference between the mixed strategy outcome and the base case is that when the quality is low, the firm's profit is independent of information precision. Finally, we investigate the effect of changes in the fraction of high precision signals. We show that the results are qualitatively similar to the base model.

The results of this study help to explain what happens in new categories over time. At first, a product is heavily discounted and consumers enjoy a low price. Improvements in information precision may lead to lower prices (this finding is independent of product quality). Even \textit{without competition}, we should see price decreases in the introductory stage of a new product. Furthermore, when a category is new, a high quality firm does not benefit from increasing the transparency of information sources and educating consumers. However, beyond a certain level of maturity, the inverse is true: information is precise and the ability of consumers to distinguish precise signals from imprecise signals is higher. Here, not only is a high quality firm more profitable, but the profit of the high quality firm increases as the fraction of the market who can assess precision grows. In other words, if "the conditions are right", a high quality firm can benefit by educating consumers, improving the transparency of information sources, and encouraging users to generate high precision signals. 
\newpage
\bibliographystyle{apalike}
\bibliography{mylib}

\newpage
 \renewcommand{\theequation}{A.\arabic{equation}}
\setcounter{equation}{0}

\section*{Appendix A}

\subsection*{Proof of Proposition \ref{no_separating}}
Suppose there is a separating equilibrium. According to Lemma \ref{lem_price_interval} and Lemma \ref{lem_pb}, the following must hold in this equilibrium:  $v_B=\tilde{p}(B)<\tilde{p}(G)\leq 1$. Here, consumers' posteriors satisfy $\mu(\tilde{p}(G),\sigma)=1$ and $\mu(\tilde{p}(B),\sigma)=0$. Thus, all consumers buy the product. But the low quality firm has an incentive to deviate to $\tilde{p}(G)$, which is a contradiction.

\subsection*{Proof of Proposition \ref{mixed}}
The firm's profit is a weighted average of profits for sophisticated and naive consumers:

 \begin{equation}
 \Pi_H(p;h,\lambda)= \lambda \begin{cases}
  p &   \text{  if   } 0 \leq p \leq p_1\\    
  \frac{1+h}{2}p & \text{  if   } p_1 \leq p\leq p_3  \\
    \frac{h}{2}p & \text{  if   } p_3 \leq p\leq p_5  \\
      0 & otherwise 
\end{cases}
 + (1-\lambda) \begin{cases}
  p &   \text{  if   }0 \leq p \leq p_2\\    
  \frac{1+2h}{4}p & \text{  if   } p_2\leq p\leq p_4  \\
    0 & otherwise
\end{cases}
 \end{equation}
Similarly, we can derive the profit for the low quality firm. Simplifying these expressions we get,
$$\Pi_H(p;h,\lambda)=\begin{cases}
p &   \text{  if   } 0 \leq p \leq p_1 \\
p (1-\frac{1-h}{2}\lambda) &   \text{  if   } p_1 \leq p \leq p_2 \\
p (\frac{1+2h}{4}+\frac{\lambda}{4}) &   \text{  if   } p_2 \leq p \leq p_3 \\
p (\frac{1+2h}{4}-\frac{\lambda}{4}) &   \text{  if   } p_3 \leq p \leq  p_4\\
p \frac{h}{2} \lambda &   \text{  if   } p_4 \leq p \leq p_5\\
0 &   otherwise
\end{cases}
 \text{ } \text{,} \text{ }  \text{ } \Pi_L(p;h,\lambda)=\begin{cases}
p &   if \text{     } 0 \leq p \leq p_1 \\
p (1-\frac{h}{2}\lambda) &   if \text{     } p_1 \leq p \leq p_2 \\
p (\frac{3-2h}{4}+\frac{\lambda}{4}) &   if \text{     } p_2 \leq p \leq p_3 \\
p (\frac{3-2h}{4}-\frac{\lambda}{4}) &   if \text{     } p_3 \leq p \leq  p_4\\
p \frac{1-h}{2} \lambda &   if \text{     } p_4 \leq p \leq p_5\\
0 &   otherwise
\end{cases}$$
 \\
 where $p_1=1-h+h v_B$, $p_2=\frac{3-2h}{4}+\frac{1+2h}{4}v_B$, $p_3=\frac{1+v_B}{2}$, $p_4=\frac{1+2h}{4}+\frac{3-2h}{4}v_B$, and $p_5=h+(1-h)v_B$. There are two parameters in the profit function. Thus, equilibrium outcomes depend on both $\lambda$ and $h$. We fix $\lambda$ and solve for the equilibrium as a function of $h$, given $\lambda$. Define $\hat{h}_1(\lambda)$, $\hat{h}_2(\lambda)$, and $\hat{h}_3(\lambda)$ as follows:
\begin{equation}
\label{omega_hat1}
    \hat{h}_1(\lambda) = max \ \Big\{h \in [0.5,1] \ \Big| \  p_1 \in \underset{p}{argmax} \  \Pi_H(p;h,\lambda)\Big\}
\end{equation}
\begin{equation}
\label{omega_hat2}
    \hat{h}_2(\lambda) = max \ \Big\{h \in [0.5,1] \ \Big| \  p_2 \in \underset{p}{argmax} \  \Pi_H(p;h,\lambda)\Big\}
\end{equation}
\begin{equation}
\label{omega_hat3}
    \hat{h}_3(\lambda) = max \ \Big\{h \in [0.5,1] \ \Big| \  p_3 \in \underset{p}{argmax} \  \Pi_H(p;h,\lambda)\Big\}
\end{equation}

Now, define $\hat{\lambda}_1$, $\hat{\lambda}_2$, and $\hat{\lambda}_3$ as follows,
\begin{equation}
\label{lambda_hat1}
    \hat{\lambda}_1 = min \ \Big\{\lambda \in (0,1) \ \Big| \ \hat{h}_2(\lambda) \neq \O \ , \  \lambda \in \underset{\lambda}{argmax} \ \hat{h}_2(\lambda) \Big\}
\end{equation}
\begin{equation}
\label{lambda_hat2}
    \hat{\lambda}_2 = min \ \Big\{\lambda \in (0,1) \ \Big| \ \hat{h}_3(\lambda) \neq \O \ , \  \lambda \in \underset{\lambda}{argmax} \ \hat{h}_3(\lambda) \Big\}
\end{equation}
\begin{equation}
\label{lambda_hat3}
    \hat{\lambda}_3 = min \ \Big\{\lambda \in (0,1) \ \Big| \ \hat{h}_1(\lambda) \neq \O \ , \  \lambda \in \underset{\lambda}{argmax} \ \hat{h}_1(\lambda) \Big\}
\end{equation}
Solving for the equilibrium, we have:
\\
1. If $0 \leq \lambda \leq \hat{\lambda}_1$ the equilibrium price is as follows:
\begin{equation}
    \tilde{p} = \begin{cases}
    p_1 & \text{ if } \frac{1}{2} \leq \hat{h}_1(\lambda) \\
    p_2 & \text{ if } \hat{h}_1(\lambda) \leq \hat{h}_2(\lambda)\\
     p_4  & \text{ if } \hat{h}_2(\lambda) \leq 1
    \end{cases}
\end{equation}
and equilibrium profits are:
\begin{equation}
    \tilde{\Pi}_H(h;\lambda)=\begin{cases}
    p_1 & \text{ if } \frac{1}{2} \leq \hat{h}_1(\lambda) \\
    p_2(1-\frac{1-h}{2}\lambda) & \text{ if } \hat{h}_1(\lambda) \leq \hat{h}_2(\lambda)\\
     p_4 (\frac{1+2h}{4}-\frac{\lambda}{4}) & \text{ if } \hat{h}_2(\lambda) \leq 1
    \end{cases}
\ \text{ , } \ 
    \tilde{\Pi}_L(h;\lambda)=\begin{cases}
    p_1 & \text{ if } \frac{1}{2} \leq \hat{h}_1(\lambda) \\
    p_2(1-\frac{h}{2}\lambda) & \text{ if } \hat{h}_1(\lambda) \leq \hat{h}_2(\lambda)\\
     p_4 (\frac{3-2h}{4}-\frac{\lambda}{4}) & \text{ if } \hat{h}_2(\lambda) \leq 1
    \end{cases}
\end{equation}

2. If $\hat{\lambda}_1 \leq \lambda \leq \hat{\lambda}_2$ the equilibrium price is as follows:
\begin{equation}
    \tilde{p} = \begin{cases}
    p_1 & \text{ if } \frac{1}{2} \leq \hat{h}_1(\lambda) \\
    p_2 & \text{ if } \hat{h}_1(\lambda) \leq \hat{h}_2(\lambda)\\
    p_3  & \text{ if } \hat{h}_2(\lambda) \leq \hat{h}_3(\lambda)\\
     p_4  & \text{ if } \hat{h}_3(\lambda) \leq 1
    \end{cases}
\end{equation}
and equilibrium profits are:
\begin{equation}
    \tilde{\Pi}_H(h;\lambda)=\begin{cases}
    p_1 & \text{ if } \frac{1}{2} \leq \hat{h}_1(\lambda) \\
    p_2(1-\frac{1-h}{2}\lambda) & \text{ if } \hat{h}_1(\lambda) \leq \hat{h}_2(\lambda)\\
    p_3 (\frac{1+2h}{4}+\frac{\lambda}{4}) & \text{ if } \hat{h}_2(\lambda) \leq \hat{h}_3(\lambda)\\
     p_4 (\frac{1+2h}{4}-\frac{\lambda}{4}) & \text{ if } \hat{h}_3(\lambda) \leq 1
    \end{cases}
\ \text{ , } \ 
    \tilde{\Pi}_L(h;\lambda)=\begin{cases}
    p_1 & \text{ if } \frac{1}{2} \leq \hat{h}_1(\lambda) \\
    p_2(1-\frac{h}{2}\lambda) & \text{ if } \hat{h}_1(\lambda) \leq \hat{h}_2(\lambda)\\
    p_3 (\frac{3-2h}{4}+\frac{\lambda}{4}) & \text{ if } \hat{h}_2(\lambda) \leq \hat{h}_3(\lambda)\\
     p_4 (\frac{3-2h}{4}-\frac{\lambda}{4}) & \text{ if } \hat{h}_3(\lambda) \leq 1
    \end{cases}
\end{equation}

3. If $\hat{\lambda}_2 \leq \lambda \leq \hat{\lambda}_3$ the equilibrium price is as follows:
\begin{equation}
    \tilde{p} = \begin{cases}
    p_1 & \text{ if } \frac{1}{2} \leq \hat{h}_1(\lambda) \\
    p_2 & \text{ if } \hat{h}_1(\lambda) \leq \hat{h}_2(\lambda)\\
     p_3  & \text{ if } \hat{h}_2(\lambda) \leq 1
    \end{cases}
\end{equation}
and equilibrium profits are:
\begin{equation}
    \tilde{\Pi}_H(h;\lambda)=\begin{cases}
    p_1 & \text{ if } \frac{1}{2} \leq \hat{h}_1(\lambda) \\
    p_2(1-\frac{1-h}{2}\lambda) & \text{ if } \hat{h}_1(\lambda) \leq \hat{h}_2(\lambda)\\
     p_3 (\frac{1+2h}{4}+\frac{\lambda}{4}) & \text{ if } \hat{h}_2(\lambda) \leq 1
    \end{cases}
\ \text{ , } \ 
    \tilde{\Pi}_L(h;\lambda)=\begin{cases}
    p_1 & \text{ if } \frac{1}{2} \leq \hat{h}_1(\lambda) \\
    p_2(1-\frac{h}{2}\lambda) & \text{ if } \hat{h}_1(\lambda) \leq \hat{h}_2(\lambda)\\
     p_3 (\frac{3-2h}{4}+\frac{\lambda}{4}) & \text{ if } \hat{h}_2(\lambda) \leq 1
    \end{cases}
\end{equation}
\\
4. If $\hat{\lambda}_3 \leq \lambda \leq 1$ the equilibrium price is as follows:
\begin{equation}
    \tilde{p} = \begin{cases}
    p_1 & \text{ if } \frac{1}{2} \leq \hat{h}_1(\lambda) \\
    p_3  & \text{ if } \hat{h}_1(\lambda) \leq 1
    \end{cases}
\end{equation}
and equilibrium profits are:
\begin{equation}
    \tilde{\Pi}_H(h;\lambda)=\begin{cases}
    p_1 & \text{ if } \frac{1}{2} \leq \hat{h}_1(\lambda) \\
    p_3 (\frac{1+2h}{4}+\frac{\lambda}{4}) & \text{ if } \hat{h}_1(\lambda) \leq 1
    \end{cases}
\ \text{ , } \ 
    \tilde{\Pi}_L(h;\lambda)=\begin{cases}
    p_1 & \text{ if } \frac{1}{2} \leq \hat{h}_1(\lambda) \\
    p_3 (\frac{3-2h}{4}+\frac{\lambda}{4}) & \text{ if } \hat{h}_1(\lambda) \leq 1
\end{cases}
\end{equation}
The final step is to complete the proof is to define $h^*(\lambda)$ and $\bar{v}$. Define,
\begin{align}
    h^*(\lambda) = 
    \begin{cases}
    \hat{h}_2(\lambda) & \text{if} \ \text{ } 0 \leq \lambda \leq \hat{\lambda}_3\\
    \hat{h}_1(\lambda) & \text{if} \ \text{ } \hat{\lambda}_3 \leq \lambda \leq 1
    \end{cases}
\end{align}
In all cases above, given $\lambda$, the equilibrium profit of the high quality firm is decreasing in $h$, when $h<h^*(\lambda)$, and it is increasing when $h>h^*(\lambda)$. 

For the low quality firm, we need to make sure that the firm does not deviate to $p=v_B$. Define $\bar{v}=\{v_B \in (0,1) \ \big| \  \underset{\lambda}{min} \ \tilde{\Pi}_L(1,\lambda)=v_B \}$. As long as, $v_B\leq \bar{v}$, by substituting the equilibrium price into the low quality firm's profit function, it is straightforward to show that for all $h \in [0.5,1]$, the profit is decreasing in $h$, given $\lambda$.

\subsection*{Proof of Corollary \ref{th_n>th_i}}
Based on the proof of Proposition \ref{mixed}, we have $h^*(0) = \frac{-3+v_B+\sqrt{(3-v_B)^2+(1-v_B)(11+v_B)}}{2(1-v_B)}$ and $h^*(1)=\frac{3-v_B}{5-v_B}$. This implies $h^*(1) < h^*(0)$ for $v_B \in (0,1)$.

\subsection*{Proof of Corollary \ref{prefer_firm}}
Denote the profit functions in a sophisticated market and a naive market by superscripts $S$ and $N$, respectively. For the low quality firm, based on the proof of Proposition \ref{mixed}, we have $\tilde{\Pi}^S_L(h)<\tilde{\Pi}^N_L(h)$ when $h < h^*(0)$, and $\tilde{\Pi}^S_L(h)\geq \tilde{\Pi}^N_L(h)$ when $h > h^*(0)$. For the high quality firm, however, there are more steps. When $h < h^*(1)$ it is obvious that $\tilde{\Pi}^S_H(h)<\tilde{\Pi}^N_H(h)$. Define $\underline{h}=\frac{2}{3-v_B}<h^*(0)$ and $\bar{h} = max \ \Big\{h \in [0.5,1]\ \Big|\ 4(1+h)(1+v_B) \geq (1+2h)(1+2h+v_B(3-2h))\Big\}>h^*(0)$. Then, when $h \in (\underline{h},\bar{h})$, we have $\tilde{\Pi}^N_H(h)\leq \tilde{\Pi}^S_H(h)$. Alternatively, when $h \in (0.5,\underline{h}) \cup (\bar{h},1)$, $\tilde{\Pi}^N_H(h)\geq \tilde{\Pi}^S_H(h)$.

\subsection*{Proof of Proposition \ref{lambda}}

Here, we build upon the proof of Proposition \ref{mixed}. In all cases, the profit of both low quality and high quality firms are independent of $\lambda$ when $h < \hat{h}_1$. So, we just examine situations where $h \geq \hat{h}_1$. In Proposition \ref{mixed} we had 4 cases:

In case 1, the profits are decreasing in $\lambda$ when $h \geq \hat{h}_1$. In case 2, the profits are decreasing in $\lambda$ when $h \in [\hat h_1, \hat h_2] \cup [\hat{h}_3,1]$, and are increasing in $\lambda$ when $h \in (\hat{h}_2,\hat{h}_3)$. In case 3, the profits decrease in $\lambda$ when $\hat h_1 \leq h \leq \hat{h}_2$, and increase in $\lambda$ when $h>\hat{h}_2$. In case 4, the profits increase in $\lambda$ for $h \geq \hat{h}_1$. Define $\hat{h}^{-1}_i$ as the inverse of $\hat{h}_i$, $i \in \{1,2,3\}$, where $\hat{h}_i$ is defined in \eqref{omega_hat1} - \eqref{omega_hat3}. Then, $\bar{\lambda}$ is as follows:
\begin{align}
    \bar{\lambda}(h) =\begin{cases}
        \hat{h}_1^{-1}(h) & \text{ if } \ 0.5 \leq h\leq \hat{h}_1(\hat{\lambda}_3) \\
        \hat{h}_2^{-1}(h) & \text{ if } \ \hat{h}_1(\hat{\lambda}_3) \leq h\leq \hat{h}_2(\hat{\lambda}_1) \\
        \hat{h}_3^{-1}(h) & \text{ if } \ \hat{h}_2(\hat{\lambda}_1) \leq h\leq \hat{h}_3(\hat{\lambda}_2)
    \end{cases}
\end{align}
where $\hat{\lambda}_i$, $i \in \{1,2,3\}$, is defined earlier in \eqref{lambda_hat1} - \eqref{lambda_hat3}.

\subsection*{Proof of Proposition \ref{prop:mixed_naive}}

Substituting $\bar{p}$ and \eqref{eqm_beliefs} in \eqref{indifference} we get,
 \begin{equation}
 \label{indifference_simplified}
     \frac{3-2h}{4}\Big[(\frac{\mu(\sigma_g)}{\mu(\sigma_g)+(1-\mu(\sigma_g))\alpha})(1-v_B)+v_B\Big]=v_B
 \end{equation}
 where $\mu(\sigma_g) = Pr(G|\sigma_g)=\frac{1+2h}{4}$. Solving for $\alpha$, equation \eqref{indifference_simplified} gives the equilibrium mixing probability of the low quality firm, $\alpha^*(h,v_B)=\frac{1}{v_B}-\frac{4}{3-2h}$. Feasibility requires that $\alpha^* \in (0,1)$, which is the case when $v_B \in (\frac{3-2h}{7-2h},\frac{3-2h}{4})$.\footnote{For instance, when $h$ is 1, there is a mixed-strategy equilibrium specified as above, for $v_B \in (0.2,0.25)$.}

\subsection*{Proof of Proposition \ref{prop:unconditional-naive}}
The firm's expected profit, given $Q$, is as follows:
 \begin{equation}
 \Pi_H(p;h,\gamma)= \begin{cases}
  p &   \text{  if   }  0 \leq p \leq p_1\\    
 \frac{1+\gamma(2h-1)}{2}p & \text{  if   } p_1\leq p\leq p_2  \\
    0 & otherwise
\end{cases}
\text{    ,     }
  \Pi_L(p;h,\gamma)= \begin{cases}
  p &   \text{  if   } 0 \leq p \leq p_1\\    
 \frac{1-\gamma(2h-1)}{2}p & \text{  if   } p_1\leq p\leq p_2  \\
    0 & otherwise
    \end{cases}
 \end{equation}
where $p_1 = 1-\bar{w}(1-v_B)=1-\bar{w}$ and $p_2=v_B+\bar{w}(1-v_B)=\bar{w}$, and $\bar{w}=\frac{1+\gamma(2h-1)}{2}$. First of all, the low quality firm's profit is always decreasing in $\gamma$, because both $p_1$ and $\frac{1-\gamma(2h-1)}{2}p_2$ are decreasing in $\gamma$. This proves part (a) of the proposition. 

Second of all, define $\underline{h}=max \ \{h \in (0.5,1) \ \big| \  1-h>h^2 \}={-1+\sqrt{5} \over 2}$ and $\bar{h}=min \ \{h \in (0.5,1) \ \big| \  {{3-2h} \over 4}<({{1+2h} \over 4})^2 \}=-1.5+\sqrt{5}$. When $h \in (0.5,\underline{h})$, the equilibrium price is $\tilde{p}=p_1$. In this case all consumers buy. The high quality firm does not deviate because $p_1>\bar{w}p_2>v_B=0$. The low quality firm also does not deviate because $p_1>(1-\bar{w})p_2>v_B=0$. Here, the high quality firm's profit is decreasing in $\gamma$ for all $\gamma \in (0,1)$. 

When $h \in (\bar{h},1)$, the equilibrium price is $\tilde{p}=p_2$. Here, only consumers who obtain a good signal buy. The high quality firm does not deviate because $\bar{w}p_2>p_1>v_B=0$. The low quality firm also does not deviate because $(1-\bar{w})p_2>v_B=0$. In this case, the high quality firm's profit is increasing in $\gamma$  for all $\gamma \in (0,1)$.

Finally, let's assume $h \in (\underline{h},\bar{h})$. Define $\hat{\gamma} = \{\gamma \in (0,1) \  \big| \ p_1 = \bar{w}p_2\}$. Here, when $\gamma \in (0,\hat{\gamma})$ the optimal price is $p_1$ and when $\gamma \in (\hat{\gamma},1)$ the optimal price is $p_2$. Thus, the equilibrium price can be summarized as follows:
 \begin{equation}
 \tilde{p}=\begin{cases}
  p_1 &   \text{  if   } \gamma \in (0,\hat{\gamma}) \\    
   p_2 &  \text{  if   } \gamma \in (\hat{\gamma},1) 
\end{cases}
 \end{equation}
 The equilibrium profits are then as follows:
  \begin{equation}
 \tilde{\Pi}_H(\gamma;h)=\begin{cases}
  p_1 &   \text{  if   } \gamma \in  (0,\hat{\gamma}) \\    
   \bar{w}p_2 &  \text{  if   }  \gamma \in (\hat{\gamma},1)
\end{cases}
\text{ } \text{    ,     } \text{ }
\tilde{\Pi}_L(\gamma;h)=\begin{cases}
  p_1 &   \text{  if   } \gamma \in  (0,\hat{\gamma}) \\    
    (1-\bar{w})p_2 &  \text{  if   }  \gamma \in (\hat{\gamma},1)
\end{cases}
\end{equation}
This implies that $\frac{\partial \tilde{\Pi}_H}{\partial \gamma}<0$  for  $\gamma<\hat{\gamma}$, and $\frac{\partial \tilde{\Pi}_H}{\partial \gamma}>0 $ for  $\gamma>\hat{\gamma}$. In addition, $\frac{\partial \tilde{\Pi}_L}{\partial \gamma}<0$ for all $\gamma$.

\subsection*{Proof of Proposition \ref{prior}}

We need to look for pooling equilibria, $\tilde{p}(G)=\tilde{p}(B)=\tilde{p}$, such that consumers' beliefs are $\mu(\sigma,p)=Pr(G|\sigma)$. Let's assume consumers are naive. The firm knows the product quality and consumers type. Thus, its expected profit, given $Q$, is as follows:\\
 \begin{equation}
 \Pi_H(p;h,\mu_0)= \begin{cases}
  p &   \text{  if   }  p \leq \underline{p}\\    
  \frac{1+2h}{4}p & \text{  if   } \underline{p}\leq p\leq \bar{p}  \\
    0 & otherwise
\end{cases}
\text{    ,     }
  \Pi_L(p;h,\mu_0)= \begin{cases}
  p &   \text{  if   } p \leq \underline{p}\\    
  \frac{3-2h}{4}p & \text{  if   } \underline{p}\leq p\leq \bar{p}  \\
    0 & otherwise
    \end{cases}
 \end{equation}
\\
where $\underline{p}=\frac{\mu_0(3-2h)+(1-\mu_0)(1+2h)v_B}{\mu_0(3-2h)+(1-\mu_0)(1+2h)}$ and $\bar{p}=\frac{\mu_0(1+2h)+(1-\mu_0)(3-2h)v_B}{\mu_0(1+2h)+(1-\mu_0)(3-2h)}$ are the two potentially optimal prices. If $p = \underline{p}$ the consumer will buy the product regardless of her signal. But, if $p= \bar{p}$ she will buy the product only if she receives a good signal, $\sigma_{g}$. Define $h^*(\mu_0)=\{h \in (0.5,1) \text{ } \ \big| \ \text{ }\underline{p}=\frac{(1+2h)}{4}\bar{p}\}$. When $0.5<h<h^*(\mu_0)$, the equilibrium price is at $p =\underline{p}$, in which all consumers buy. If a firm deviates to $p=\bar{p}$, although it sells at a higher price, only consumers who obtain a good signal buy the product. It is then straightforward to show that this leads to a lower profit. Thus, no firm has incentives to deviate from $p = \underline{p}$. 

When $h^*(\mu_0)<h<1$, the equilibrium price is at $p = \bar{p}$. The high quality firm does not have incentives to deviate, because $\frac{(1+2h)}{4}\bar{p}>\underline{p}$. 
 Therefore, the equilibrium price can be summarized as follows:
 \begin{equation}
 \tilde{p}=\begin{cases}
  \underline{p} &   \text{  if   } h\in (0.5,h^*(\mu_0)) \\    
   \bar{p} &  \text{  if   }h\in (h^*(\mu_0),1) 
\end{cases}
 \end{equation}
and the equilibrium profits are:
 \begin{equation}
 \tilde{\Pi}_H(h;\mu_0)=\begin{cases}
  \underline{p} &   \text{  if   } h\in (0.5,h^*(\mu_0)) \\    
    \frac{(1+2h)}{4}\bar{p} &  \text{  if   }h\in (h^*(\mu_0),1)
\end{cases}
\text{   ,  }
\tilde{\Pi}_L(h;\mu_0)=\begin{cases}
  \underline{p} &   \text{  if   }h\in (0.5,h^*(\mu_0)) \\    
    \frac{(3-2h)}{4}\bar{p} &  \text{  if   }h\in (h^*(\mu_0),1)
\end{cases}
\end{equation}
This implies that $\frac{\partial \tilde{\Pi}_H}{\partial h}\leq 0$  for  $h<h^*(\mu_0)$, and $\frac{\partial \tilde{\Pi}_H}{\partial h}\geq 0 $ for  $h>h^*(\mu_0)$. 

For the low quality firm, we need to make sure that the firm does not deviate to $p=v_B$. Define, $\underline{\mu}=min \ \{\mu_0 \in (0,0.5)\text{ }|\text{ }\forall \text{ } h \in (h^*(\mu_0),1): \frac{\partial}{\partial h}(\frac{(3-2h)\bar{p}}{4})<0\}$, and $\bar{v}=\{v_B \in (0,1) \ \big| \  \underset{\mu_0 \in (\underline{\mu},1)}{min} \ \tilde{\Pi}_L(1,\mu_0)=v_B \}$. As long as, $v_B\leq \bar{v}$ and $\mu_0>\underline{\mu}$, by substituting the equilibrium price into the low quality firm's profit function, one can easily conclude that the profit is decreasing in $h$ for all $h \in [0.5,1]$, given $\lambda$.
\\

 \newpage
 
\renewcommand{\theequation}{B.\arabic{equation}}
\renewcommand{\theprop}{B.\arabic{prop}}
\renewcommand{\thefigure}{B.\arabic{figure}}
\setcounter{equation}{0}
\setcounter{prop}{0}
\setcounter{figure}{0}

\end{document}